%% file: main.tex
\documentclass[twoside]{article} 
\usepackage{amssymb,amsmath,amsthm,mathrsfs,algorithm,algorithmic, xcolor, soul}
\usepackage{latexsym,graphicx,url,enumerate,hyperref, comment, adjustbox, nicematrix}
\usepackage{PRIMEarxiv}

\newtheorem{prop}{Proposition}

\usepackage{xr}

\input{abbreviation}
\date{}

\title{}
\title{Achieving Privacy Utility Balance for Multivariate Time Series Data
}

\author{
  Gaurab Hore  \\
	University of Maryland Baltimore County\\
    1000 Hilltop Cir, Baltimore, MD 21250\\
	\texttt{gaurabh1@umbc.edu} \\
   \And
  Tucker McElroy \\
	Research and Methodology Directorate, U.S. Census Bureau\\
	4600 Silver Hill Road,Washington, D.C. 20233-9100, USA\\
	\texttt{tucker.s.mcelroy@census.gov}\\
    \And
    Anindya Roy \\
	University of Maryland Baltimore County\\
    1000 Hilltop Cir, Baltimore, MD 21250\\
	\texttt{anindya@umbc.edu} 
}

\begin{document}
\maketitle
\begin{abstract}
 Utility-preserving data privatization is of utmost importance for data-producing agencies. 
 The popular noise-addition privacy
 mechanism distorts autocorrelation patterns in time series data, thereby marring utility;  in response, \cite{MRH2023} introduced all-pass filtering  (FLIP) as a utility-preserving time series data privatization method. Adapting this concept to multivariate data is more complex, and in this paper we propose a multivariate all-pass (MAP) filtering method, employing an optimization algorithm to achieve the best balance between data utility and privacy protection. To test the effectiveness of our approach, we apply MAP filtering to both simulated and real data, sourced from the U.S. Census Bureau's Quarterly Workforce Indicator (QWI) dataset.
\end{abstract}

\paragraph{Disclaimer}
Any opinions and conclusions expressed herein are those of the authors and do not represent the views of the U.S. Census Bureau. All results in this paper use publicly available data from Census Bureau websites.

\section{Introduction}

With increased digital participation and online engagement, safeguarding sensitive data has become extremely important over the last decade. Researchers have devised innovative approaches for data privacy and a multitude of privacy measures along with their implementation mechanisms have been proposed in the literature. However, most privacy mechanisms rely upon noise addition or noise multiplication methods. 
When it comes to time series, noise addition (or multiplication) may significantly change the autocorrelation structure, thereby diminishing the quality and utility of such data. Also, since the privacy measures are developed primarily for databases with independent entries, the privacy guarantees no longer hold for dependent data (such as time series data). 

Thus, there is a critical need for privacy mechanisms for time series data that ensure both privacy and data utility. In \cite{WaZh2009}, the authors forcefully argue for maintaining data utility while implementing disclosure avoidance algorithms. We concur with that sentiment. However, there is a shortage of privacy mechanisms for time series data that maintain data utility -- methods focused on privacy fail to take properties of the temporal dynamics (such as serial correlation) into account. In a recent paper, \cite{MRH2023}  proposed a proper privacy-utility framework called FLIP for regularly-spaced time series data. There are several other approaches for univariate time series, developed in different disciplines like economics, cryptography, data mining, and data-streaming (and under different engineering applications such as power-grid) that are available in the literature; see \cite{Abowd2012}, \cite{RN2010}, \cite{Acs2012}, \cite{Lyu2017}, 
\cite{Lako2021}, \cite{Hong2013}, \cite{SCR2011}, \cite{EFM2015}, \cite{SST2009}, \cite{Fior2019}, \cite{Stach2019}, \cite{Kats2022}, and the references therein. However, these approaches do not use any mathematical framework for optimizing the privacy-utility trade-off. 

For multiple time series, the need to account for utility while ensuring privacy is even more stark. This is because most approaches perform privacy evaluation on a series-by-series basis and ignore the cross-correlation structure, a critical component of data utility for multivariate time series. To our knowledge, there are no privacy procedures that preserve cross-series dependence information along with marginal time series properties. The present article fills that gap. 
Recently, several articles have looked at forecasting properties of multiple time series after the application of privacy mechanisms. Many are based on deep-learning and predictive structures for dynamical models; see \cite{Arcolezi2022}, \cite{Imtiaz2020}, \cite{Leukam2021}. By examining the forecast properties of privatized series, these approaches do consider data utility, but they do not use any formal framework for balancing privacy and utility. Overall, it seems that privacy mechanisms that formally address data utility are not available for multiple time series. This article proposes a formal privacy utility framework for regularly-spaced multiple time series. 

The FLIP methodology suggested by \cite{MRH2023}  employs all-pass filtering to achieve privacy while simultaneously preserving time series data utility. Focusing on regularly-sampled time series data, \cite{MRH2023} introduced the ``Linear Incremental Privacy" (LIP) measure, and presented a novel set of utility conditions referred to as ``second-order utility". Although a predictive measure that is more apt for time series, the incremental privacy addressed in the LIP framework is similar in spirit to {\it differential privacy} (DP), which addresses disclosure avoidance beyond what is available to the attacker.  Differential privacy is a popular privacy measure
that provides a formal mathematical definition of privacy. Developed in a series of papers (\cite{Dw2006}, \cite{DwMcNiSm2006}, \cite{DwRo2014}),  DP is generally accepted as the gold standard data privacy measure. It has been accepted widely in industrial and government data protection plans, including implementation in the decennial U.S. census, probably one of the largest and most complex data collection exercises; see \cite{Abowd2020}. One of the reasons for their popularity is that DP mechanisms provide hard privacy guarantees. An elegant statistical interpretation of DP is given in \cite{Meng2020}, where the authors establish the concept of DP in terms of a posterior quantity, making it attractive for statistical maneuvering. Despite this, DP and other popular data perturbation tools are primarily designed for databases with independent entries
(i.e.,  the mathematical formulation is valid under the independence assumption), and lack optimality properties for time series data. While some articles (\cite{SoWaCh2017} and \cite{SoCh2017}) examined modified DP mechanisms that are applied to time series structures,  none provide any optimal balancing of privacy and utility.

Whereas the incremental privacy measure under LIP can be extended to the multivariate time series context, the concept of all-pass filtering that was the primary tool for implementation of LIP is non-trivial to formulate in the multivariate case.  In particular, the filters become matrix-valued, and hence the algebra is no longer commutative,
 complicating the mathematics. The main goal of this article is to develop the multivariate generalization of FLIP along with a generalization of all-pass filtering for multiple time series. The multivariate 
all-pass filter is developed in Section ~2 and the multivariate extension of LIP, called $\mlip$ is discussed in Section~3. Section~4 provides the details for the implementation of $\mlip$ in practice. Limited numerical studies are given in Section~5 along with a real data application that examines Quarterly Workforce Indicator (QWI) data published by the U.S. Census Bureau. Section~6 provides conclusions and a discussion of future work.

\section{Multivariate All-pass Filters}
\lb{All-pass filtering}
The primary mathematical concept used in the development of the FLIP mechanism (\cite{MRH2023}) for univariate time series was the idea of all-pass filtering; here we make non-trivial extensions
to the multivariate setting.
In particular, we define the notion of a multivariate all-pass filter and describe a suitable filter class   that is particularly suitable for the privacy application. Hereafter, we employ the following notation:  the braces notation $\{\mathbf{X}_t\}$ (the bold font indicates that 
the time series is multivariate) denotes the entire time series, while $\mathbf{X}_t$  denotes the single random vector at time $t$.

\subsection{Multivariate All-Pass Filtering}

Suppose that  $\{ \mathbf{X}_t \}$  is a second-order stationary multivariate time series of dimension $n$, with components denoted
 by ${X}_{j,t}$ for $1 \leq j \leq n$.  Denoting the process' autocovariance function by $\Gamma_{ \mathbf{X} } (h) = \mbox{Cov} ( \mathbf{X}_{t+h}, \mathbf{X}_t )$ for $h \in \mathrm{Z}$,  its spectral density is  defined by  $S_{\mathbf{X}} (\lambda) =  \sum_h e^{-i h \lambda}  \Gamma_{ \mathbf{X} } (h)$ for $\lambda \in [-\pi, \pi]$. Suppose for a complex matrix $\bA$, $\bA^*$ denotes its conjugate transpose. Then $S_{\bX}$ is a matrix-valued function from $[-\pi, \pi]$  to $\CC^{n\times n}$ such that, for each $\lambda \in [-\pi, \pi]$, $S_{\bX}(\lambda)$ is a hermitian non-negative definite matrix and $S_{\mathbf{X}} (-\lambda) = S^*_{\mathbf{X}} (\lambda).$ 
 
 
\begin{Definition}
Let $S: [-\pi, \pi] \to \CC^{n\times n}$ be such that for each $\lambda \in [-\pi, \pi]$, $S(\lambda)$ is a hermitian non-negative definite matrix and $S(-\lambda) = S^*(\lambda).$ Then $S$ will be called a {\it spectral density matrix function}, or simply a {\it spectral density}. 
\lb{def:spec_den}
\end{Definition}
 
We will use the abbreviation  $z = e^{-i \lambda}$, so that $S_{\mathbf{X}} (\lambda) =  \sum_h z^h \Gamma_{ \mathbf{X} } (h)$.
The autocovariances can be recovered from the spectral density via the formula 
\[
    \Gamma_{ \mathbf{X} } (h) = 
    \frac{1}{2 \pi} \int_{-\pi}^{\pi}
    e^{i \lambda h} S_{\mathbf{X}} (\lambda) d\lambda.
\]
 Letting $B$ denote the backshift operator \cite{MP2020}, $\Psi (B) =\sum_k \Psi_k B^k $  defines a multivariate linear time-invariant filter, where each coefficient $\Psi_k$ is a $n \times n$-dimensional matrix.  This filter operates on a time series  $\{ \mathbf{X}_t \}$
as follows:
\[
\Psi (B) \mathbf{X}_t = 
\sum_k \Psi_k \mathbf{X}_{t-k}.
\]
Evaluating the filter at $z = e^{-i \lambda}$ yields the frequency response function of the  filter, viz. $\Psi (z) = \sum_k \Psi_k z^k$. Letting  $\{ \mathbf{Y}_t \}$ denote this filter output, it is also second-order stationary so long as the filter's frequency response function has finite matrix norm at each $\lambda$.  Then 
$S_{\mathbf{Y}}$ is related to $S_{\mathbf{X}}$
 as (see \cite{Brillinger})
 \begin{equation}
 S_{\mathbf{Y}}(\lambda)=\Psi(z)S_{\mathbf{X}}(\lambda)  { \Psi(z^{-1}) }^{\prime}. 
 \label{eq:allpass}
 \end{equation}
When $n=1$ (the univariate case), $\Psi (z)$
is an all-pass filter if $|\Psi (z)|=1$ for
all $\lambda$, and hence $S_{\mathbf{Y}} \equiv
S_{\mathbf{X}}$.  Extending this concept to
the multivariate context ($n > 1$),
we say that a matrix filter $\Psi (z)$ is all-pass
if $S_{\mathbf{Y}} \equiv S_{\mathbf{X}}$
in (\ref{eq:allpass}).  Though we might conjecture
that it is sufficient to demand that 
$\Psi (z)$ is unitary (i.e., $\Psi (z) {\Psi (z) }^* = \bI$, the identity matrix) for each $\lambda$, such a condition is too demanding in
practice;
for the relation \eqref{eq:allpass} to hold with $S_{\bY} = S_{\bX}$ and  for any spectral density $S_{\bX}$, $\Psi(z)$ must commute with every spectral density matrix function (of the same order) at each frequency $\lambda$.  This occurs if and only if $\Psi (z) = \bI$. Thus, there are no universal all-pass filters in the matrix case other than the trivial identity filter. 

 Fortunately, for the data privacy application we only need to filter specific series whose spectral density is known to the data curator. Thus, it suffices to generate a class of filters that act as all-pass filters for a given spectral density $S_{\bX}.$  Given this background, we can state the definition of the desired multivariate all-pass filter for a specified spectral density $S$ as the following.

\begin{Definition}[$S$-Multivariate All-Pass or $\smap$]
\label{def:mvar-allpass}
Given a spectral density matrix function $S$, 
a linear time invariant filter $\Psi (B)$ is
said to be {\it $S$-Multivariate All-Pass } (or $\smap$ for short)
if the relation 
\[ S(\lambda) = \Psi(z)S(\lambda)  { \Psi(z) }^{*} \] 
holds for all $\lambda \in [-\pi, \, \pi].$
 \end{Definition} 

 In view of  Definition~\ref{def:mvar-allpass} and equation \eqref{eq:allpass}, 
 if $\{ \bX_t \}$ is a second-order stationary time series with spectral density $S_{\bX}$, and if $\bY_t = \Psi(B)\bX_t$ is the filtered series, then the spectral density $S_{\mathbf{Y}}$ of $\{ \bY_t \}$ equals $S_{\mathbf{X}}$ provided $\Psi$ is $\sxmap.$ If $\Psi$ is $\sxmap$, then it implies that  the autocovariances of $\{ \mathbf{Y}_t \}$ are
 the same as those of $\{ \mathbf{X}_t \}$.
Clearly, given an $n-$dimensional spectral density $S_{\bX}$, $\Psi (z) = \bI$  
is a trivial $\sxmap$ filter,
 but there are many more choices.

\subsection{ A Class of Multivariate All-pass Filters}
To generate $\smap$ filters, one needs to find $\Psi$ that satisfies the condition in Definition \ref{def:mvar-allpass}. However, directly solving for the components of the filter from the equation in Definition~\ref{def:mvar-allpass} is not a feasible option. Consider the $n=2$ case.
The spectral density   for 
$\{ \mathbf{X}_t \}$ is a $2 \times 2$ matrix for
each frequency $\lambda$, and we denote the four
scalar entries as follows:
\[
S_\mathbf{X}(\lambda)=\begin{pmatrix}
    S_{X_1}(\lambda) & S_{X_1X_2}(\lambda) \\
    S_{X_2X_1}(\lambda) & S_{X_2}(\lambda)
\end{pmatrix}.
\]
Similarly, a bivariate filter can be written in
terms of scalar filters as follows:
\[
\Psi(z)=\begin{pmatrix}
    \Psi_{11}(z) & \Psi_{12} (z) \\
    \Psi_{21} (z) & \Psi_{22} (z)
\end{pmatrix}.
\]
The filtered series $\{\mathbf{Y}_t\}$ 
can then be expressed via
\[
\mathbf{Y}_t= \Psi(B)\mathbf{X}_t= \begin{pmatrix}
    \Psi_{11} (B) & \Psi_{12} (B) \\
    \Psi_{21} (B) & \Psi_{22} (B)
\end{pmatrix} \begin{pmatrix} X_{1,t}\\ X_{2,t}\end{pmatrix} = \begin{pmatrix}
    \Psi_{11} (B) X_{1,t} +  \Psi_{12} (B) X_{2,t} \\ \Psi_{21} (B) X_{1,t} +  \Psi_{22} (B) X_{2,t}
\end{pmatrix}.
\]
The spectral density matrix for $\{\mathbf{Y}_t\}$ 
is given by (\ref{eq:allpass}),
and by setting $S_\mathbf{Y} = S_\mathbf{X}$,
we can proceed to determine the scalar filters
$\Psi_{11} (B)$, $\Psi_{21} (B)$, $\Psi_{12} (B)$,
and $\Psi_{22} (B)$. When $n > 2$ there will be $n^2$
functions $\Psi_{jk} (B)$ to determine such that
 $S_\mathbf{Y} \equiv S_\mathbf{X}$, and the task of finding such solutions becomes formidable even for small to moderate $n$. It will be advantageous to find suitable special cases for which closed-form solutions are readily available.

 We next develop
 a special case that will be useful in our
 more general treatment. Suppose that
 $\{ \mathbf{X}_t \}$ is a white noise time series
 of covariance matrix $\bI$, so that
  $S_\mathbf{X} (\lambda ) = \bI$.
  Then the all-pass condition becomes
\begin{equation}
    \label{eq:allpass-condition-id}
 \bI = \Psi (z) { \Psi (z^{-1}) }^{\prime}
\end{equation}
for $z = e^{-i \lambda}$, and all $\lambda \in [-\pi, \pi]$ (i.e., $\Psi (z)$ is unitary for all $\lambda$). One way to parameterize such unitary functions is through the matrix {\it cepstral} representation discussed in \cite{holan2017cepstral}.
 Consider a matrix Laurent series 
$\Omega(z)  = \sum_{k \in \mathrm{Z}}
\Omega_k z^k$
that is related to $\Psi (z)$ via
\be
\Psi(z)  = \exp \{ \Omega (z) \}.
\lb{cepstral unitary}
\ee
Then $\Omega (z)$ is the cepstral representation
of $\Psi (z)$, and the  $\Omega_k$ are the matrix cepstral coefficients.
Then (\ref{eq:allpass-condition-id}) implies that
\[
 \bI = \exp \{ \Omega (z) \}
  \exp \{ {\Omega (z^{-1}) }^{\prime} \},
\]
using the transpose property of the matrix 
exponential.  Recall that $z = e^{-i \lambda}$,
so $z^{-1} = e^{i \lambda} = \bar{z}$.
If $\Omega (z) = - {\Omega (\bar{z})}^{\prime}$,
then (since $\Omega (z)$ and $-\Omega (z)$ commute)
\[
  \exp \{ \Omega (z) \}
  \exp \{ - \Omega (z) \}
  = \exp \{ \Omega (z) - \Omega (z) \}
  = \exp \{ 0 \} = \bI.
\]
This condition on $\Omega (z)$ means
that $\Omega_{k} = - \Omega_{-k}^{\prime}$ 
for $k \in \mathrm{Z}$, 
implying $\Omega_0$ is a skew-symmetric matrix. 
We let $\mathcal{S}_n$ denote the set of
real $n$-dimensional skew-symmetric matrices.
Hence, anti-symmetric cepstral coefficients correspond to a unitary filter $\Psi (z)$.

We will use the parameterization of the unitary operators in terms of its cepstral representation to generate a suitable parametric class of $\smap$ filters for any specified spectral density $S$. 
For developing the special case of $\smap$ filters, we will assume  \\

\noindent {\bf Assumption~PD:} For each $\lambda \in [-\pi, \pi],$ the spectral density matrix $S(\lambda)$ is positive definite.\\

\noindent Also, we will use the following result, whose straightforward proof is omitted.
\begin{Result}
    Let $\bA$ and $\bB$ be two $n\times n$ complex nonsingular matrices. Then $\bA\bA^* = \bB\bB^*$ if and only if there exists a unitary matrix $\bU$ such that $\bA\bU = \bB.$
    \lb{lem:sqrt_equiv}
\end{Result}
\noindent 
Assumption {\bf PD}  states that the multiple time series to be protected are not cointegrated in the frequency domain at particular frequencies. From an implementation point of view, the assumption is not restrictive since under numerical estimation of the spectral density of the sensitive series, the estimate can be constrained to satisfy the assumption. 

Suppose a spectral density $S_{\bX}(\lambda)$ is given, and it is assumed to be positive definite at each $\lambda$. 
Under the positive definiteness assumption, at each frequency $\lambda \in [-\pi, \pi]$, the spectral density matrix $S_{\bX}(\lambda)$ admits a non-singular square root $S_{\bX}^+(\lambda)$, i.e., for each $\lambda \in [-\pi, \pi]$ we can find a full rank matrix $S_{\bX}^+(\lambda)$  such that 
\[  
S_{\bX}(\lambda) =  S_{\bX}^+(\lambda)S_{\bX}^+(\lambda)^*. 
\]
If the filter $\Psi(z)$ is  also non-singular, then by the relation \eqref{eq:allpass}, $S_{\bY}(\lambda)$ is also positive definite at each frequency, and hence admits non-singular square roots $S_{\bY}^+(\lambda).$
Thus 
\[ S_{\bY}^+(\lambda)S_{\bY}^+(\lambda)^* = \Psi(z)S_{\bX}^+(\lambda)S_{\bX}^+(\lambda)\Psi(z)^*.\]
For $\Psi(z)$ to be $\sxmap$, a sufficient condition is $S_{\bY}^+(\lambda) = S_{\bX}^+(\lambda)$ for all $\lambda \in [-\pi, \pi].$ Hence 
\[ S_{\bX}^+(\lambda)S_{\bX}^+(\lambda)^* = (\Psi(z)S_{\bX}^+(\lambda))(\Psi(z)S_{\bX}^+(\lambda))^*.\]
Then by Result~\ref{lem:sqrt_equiv}, we have 
$S_{\bX}^+(\lambda)U(z) = \Psi(z)S_{\bX}^+(\lambda)$
for some unitary matrix $U(z).$ This implies that  $\Psi(z) = S_{\bX}^+(\lambda)U(z)S_{\bX}^+(\lambda)^{-1}.$
Thus, for a given spectral density $S$, a class of $\smap$ filters is given by 
\be
\Psi(z) = S^+(\lambda) \, U(z) \, S^+(\lambda)^{-1}.
\label{eq:map_filter}
\ee
The implications of \eqref{eq:map_filter} are substantial. It means that given a spectral density $S$, we could select the desired all-pass filters from a rich class of  $\smap$ filters, obtained by rotating the expression in \eqref{eq:map_filter} over the unitary group, and everything can be computed in closed-form. This provides flexibility in the selection of the privacy mechanism while optimizing privacy measures to attain a privacy-utility balance. 

Based on the parameterization of the unitary operator through the cepstral representation, a general class of $\smap$ filters for a given 
$n$-dimensional positive definite spectral density function $S$ can thus be defined as 
\be
\mcF_{S} = \{ S^+(\lambda) \, U(z) \, S^+(\lambda)^{-1} : U(z) = \exp \{ \sum_{k \in \ZZ} \Omega_k z^k \}, \; \Omega_k \in  \mathcal{S}_n \},
\lb{eq:s_map}
\ee
where $S^+(\lambda)$ is a square root of $S(\lambda)$ for each $\lambda \in [-\pi, \pi].$

\section{Privacy vs Utility for Multiple  Time Series}

The objective of a privacy mechanism is to transform a sensitive time series so as to mitigate disclosure risk, while also preserving its utility. In alignment with the approach presented in \cite{MRH2023}, we operate under the assumption that potential adversaries possess prior information about the sensitive series in question. We denote the sensitive series 
requiring protection as $\{\tilde{\mathbf{X}}_t\}$, and introduce auxiliary time series $\{\tilde{\mathbf{Z}}_t\}$ that encapsulate any knowledge that advanced attackers could employ to forecast the observed series.  Each of these time series -- the sensitive and the auxiliary -- are multivariate
of possibly different dimension, and has
a time-varying mean function.  We write the
de-meaned processes without a tilde, i.e.,
\be
\begin{pmatrix} \tilde{\mathbf{X}}_t\\\tilde{\mathbf{Z}}_t\end{pmatrix} = \begin{pmatrix} \mu^\mathbf{X}_t \\ \mu^\mathbf{Z}_t\end{pmatrix} + \begin{pmatrix} \mathbf{X}_t \\\mathbf{Z}_t\end{pmatrix},
\lb{eq:obs_model}
\ee
where $\{ \mathbf{X}_t, \mathbf{Z}_t \}$ are jointly stationary with spectral density matrix 
\be
 S_{\mathbf{X},\mathbf{Z}}(\lambda) = \begin{pmatrix}
S_\mathbf{X}(\lambda) & S_\mathbf{XZ}(\lambda)\\
S_\mathbf{ZX}(\lambda) & S_\mathbf{Z}(\lambda)
\end{pmatrix},
\lb{eq:specmat}
\ee
and $\{ \mu^\mathbf{X}_t, \mu^\mathbf{Z}_t \}$ are the deterministic time-varying mean functions.  We assume that these mean functions
are interpretable as trend components, and
can be   represented by deterministic functions in $t$.  Above, we use the notation 
$S_{\mathbf{X},\mathbf{Z}}$ to denote the
joint spectral density of 
$\{ \mathbf{X}_t \}$ and $\{ \mathbf{Z}_t \}$,
whereas $S_\mathbf{XZ}(\lambda)$ is their
cross-spectral density, i.e.,
$S_\mathbf{XZ}(\lambda)= \sum_{h \in \mathrm{Z}} 
 e^{-i h \lambda } \Gamma_{\mathbf{XZ}}(h)$
 for $\lambda \in [-\pi, \pi]$, where
$\Gamma_{\mathbf{XZ}}(h)= \mbox{Cov} (\mathbf{X}_{t+h},\mathbf{Z}_{t})= E(\mathbf{X}_{t+h} {\mathbf{Z}_{t}}^{\prime})$
are the cross-covariances of 
  $\{ \mathbf{X}_t \}$ and $\{ \mathbf{Z}_t \}$.

\subsection{Second-Order Utility}
We suppose that the spectral matrix $S_\mathbf{X,Z}$ is well-known to both the data-publishing agency and potential adversaries engaged in what we term an ``augury'' attack. This scenario represents an idealized context for attackers, characterized by an external source of information $\{\mathbf{Z}_t\}$. The publishing agency applies some ``privacy  mechanism'' to $\{\mathbf{X}_t\}$, thereby  producing $\{\mathbf{Y}_t\}$, which is viewed  as a proxy for the sensitive data that 
 preserves some features of interest.  
 The preservation of the autocorrelation structure of $\{\mathbf{X}_t\}$ is
 referred to as {\it second-order utility},
 and mathematically is the requirement that
 $\Gamma_{\mathbf{X}} (h) = \Gamma_{\mathbf{Y}} (h)$ for all $h  \in \mathrm{Z}$.
 This is equivalent to the requirement
 that $S_\mathbf{Y} = S_\mathbf{X}$;
 clearly, one such privacy mechanism that
 preserves second-order utility is
 all-pass filtering via $\sxmap$ filters.

\subsection{Multivariate Linear Incremental Privacy (m-LIP)}

In this subsection we formally develop our
measure of privacy.  We employ the 
following notation:  $\{\mathbf{Z}_t\}$   denotes the stationary time series of   auxiliary information, and $\mathbf{Z} = (\mathbf{Z}_1, \ldots, \mathbf{Z}_T)^{\prime}$   denotes the vector of the attacker's knowledge over the observation period $1, 2, \ldots, T$. 
We denote the average integral over $[-\pi, \pi]$ of frequency-domain functions $u$ and $S$
via $\langle u, S\rangle =
{(2 \pi)}^{-1} \int_{-\pi}^{\pi} u(\lambda) S(\lambda)^* d \lambda$. When $u(\lambda) = 1$, we simply denote the average as $\langle S\rangle$. 

Consider  a scenario where we have random vectors $\mathbf{X}$, $\mathbf{Y}$, and $\mathbf{Z}$. In the context of minimizing mean squared error (MSE) loss, the best estimate of $\mathbf{X}$ given the attacker's information $\mathbf{Z}$ is the conditional expectation denoted as $E[\mathbf{X} \vert \mathbf{Z}]$. If we publish $\mathbf{Y}$, then an updated attack that incorporates the additional information from $\mathbf{Y}$ can be expressed as $E[\mathbf{X} \vert \mathbf{Y}, \mathbf{Z}]$. For linear estimators
(which are conditional expectations if
the random vectors are jointly Gaussian), this update takes the form:
\[
E[\mathbf{X} \vert \mathbf{Y}, \mathbf{Z}]  = E[\mathbf{X} \vert \mathbf{Z}] + \mbox{Cov} [ \mathbf{X}, \mathbf{Y} \vert \mathbf{Z} ] \, { \mbox{Var} [ \mathbf{Y} \vert \mathbf{Z} ] }^{-1} \, (\mathbf{Y} - E [ \mathbf{Y} \vert \mathbf{Z} ]).
\]
The second term on the right accounts for the update to the attack resulting from the publication of $\mathbf{Y}$. We classify $\mathbf{Y}$ as ``private'' if this update equals zero for all variables $\mathbf{Z}$; in such cases, the release of $\mathbf{Y}$ does not aid the attacker in predicting $\mathbf{X}$. Calculating the MSE, we find:
\beqn \lb{main eqn}
\mbox{Var} [ \mathbf{X} \vert \mathbf{Z} ] - \mbox{Var} [ \mathbf{X} \vert \mathbf{Y}, \mathbf{Z}]
   =   \mbox{Cov} [ \mathbf{X}, \mathbf{Y} \vert \mathbf{Z} ] \, { \mbox{Var} [\mathbf{Y} \vert \mathbf{Z} ] }^{-1} \,  \mbox{Cov} [\mathbf{Y}, \mathbf{X} \vert \mathbf{Z} ].
\eeqn
Here, the left-hand side represents conditional variances of prediction of $\bX$ before and after the publication of $\mathbf{Y}$, with the difference indicating incremental vulnerability to the sensitive data. The right-hand side involves a non-negative definite matrix; this quantity equals zero when $\mathbf{Y}$ offers no assistance to the attack. Moreover, manipulation of (\ref{main eqn})
shows that $\mbox{Var} [ \mathbf{X} \vert \mathbf{Y}, \mathbf{Z}]$ is composed of the
block entries of the matrix $\mbox{Var} [ \mathbf{X}, \mathbf{Y} \vert \mathbf{Z} ]$.
In particular, $\mbox{Var} [ \mathbf{X} \vert \mathbf{Y}, \mathbf{Z}]$ is the Schur 
complement of $\mbox{Var} [ \mathbf{X}, \mathbf{Y} \vert \mathbf{Z} ]$, and hence is
itself non-negative definite.  From this fact,
it follows that
\[
\det \mbox{Var} [ \mathbf{X} \vert \mathbf{Z} ]
\geq \det \left[ \mbox{Cov} [ \mathbf{X}, \mathbf{Y} \vert \mathbf{Z} ] \, { \mbox{Var} [\mathbf{Y} \vert \mathbf{Z} ] }^{-1} \,  \mbox{Cov} [ \mathbf{Y}, \mathbf{X} \vert \mathbf{Z}] \right],
\]
 which in turn  motivates the following 
 definition of  ``privacy measure'':
\be \mathcal{P}(\mathbf{X},\mathbf{Y},\mathbf{Z}) 
= 1 -  \frac{ \det \left[ \mbox{Cov} [ \mathbf{X}, \mathbf{Y} \vert \mathbf{Z} ] \, { \mbox{Var} [\mathbf{Y} \vert \mathbf{Z} ] }^{-1} \,  \mbox{Cov} [ \mathbf{Y}, \mathbf{X} \vert \mathbf{Z}] \right]}{ \det  \mbox{Var} [ \mathbf{X} \vert \mathbf{Z} ] }.
\lb{eq:prvcy}
\ee
  The preceding discussion shows that
  the privacy measure takes values in $[0,1]$,
  and is well-defined unless $\det \mbox{Var} [ \mathbf{X} \vert \mathbf{Z} ] = 0$, which corresponds to a trivial case where the attacker already possesses the sensitive information, making privacy unattainable. Otherwise, this measure can be viewed as one minus a function of the multivariate squared conditional correlation, analogous to the familiar $R^2$ statistic from linear models.

The definition (\ref{eq:prvcy}) is 
appropriate for random vectors, or finite samples of multivariate time series,
 but we wish to develop a privacy measure for 
time series processes (irrespective of
sample size).  Next, we 
formulate a result analogous to 
(\ref{main eqn}) for stationary time series.
To do so, we focus on the sensitivity of
$\mathbf{X}_t$ given the new information 
$\mathbf{Y}_t$ over the available information
$\{ \mathbf{Z}_t \}$.  It is easy to show,
similar to (\ref{main eqn}), that
\[
\Delta \mbox{VAR}
   =   \mbox{Cov} [ \mathbf{X}_t, 
   \mathbf{Y}_t \vert \{ \mathbf{Z}_t \} ] 
   \, { \mbox{Var} [\mathbf{Y}_t \vert 
   \{ \mathbf{Z}_t \} ] }^{-1} \,  
   \mbox{Cov} [\mathbf{Y}_t, \mathbf{X}_t
   \vert \{ \mathbf{Z}_t \} ],
\]
where by definition $\Delta \mbox{VAR} =  \mbox{Var} [ \mathbf{X}_t \vert 
 \{ \mathbf{Z}_t \} ] 
 - \mbox{Var} [ \mathbf{X}_t \vert \mathbf{Y}_t, \{ \mathbf{Z}_t \} ]$ is the reduction in the conditional variance matrix from the added knowledge of the released series. 
The following result provides formulas
for these conditional variances and 
covariances, and provides the basis
for a privacy measure for stochastic processes that takes values in $[0,1]$.

\begin{prop} 
 \lb{prop:prvcy_measure}
  Let $\{\mathbf{X}_t\}$, $\{ \mathbf{Y}_t \}$, 
 and $ \{ \mathbf{Z}_t\}$ be weakly stationary 
 multivariate time series that are also
 jointly weakly stationary, where the cross-spectral densities are $S_{\mathbf{X} \mathbf{Y} }$, $S_{\mathbf{X} \mathbf{Z} }$,
 and $S_{\mathbf{Y} \mathbf{Z} }$.
Further, define the conditional spectral densities via 
\begin{align*}
   S_{ \mathbf{X} |\mathbf{Z}}
    & = S_{  \mathbf{X} }
    - S_{\mathbf{X} \mathbf{Z} }
    { S_{\mathbf{Z}}}^{-1} 
    S_{\mathbf{Z}  \mathbf{X}}, \\
   S_{ \mathbf{Y} |\mathbf{Z}}
    & = S_{  \mathbf{Y} }
    - S_{\mathbf{Y} \mathbf{Z} }
    { S_{\mathbf{Z}}}^{-1} 
    S_{\mathbf{Z}  \mathbf{Y}}, \\
    S_{\mathbf{X} \mathbf{Y} |\mathbf{Z}}
    & = S_{\mathbf{X} \mathbf{Y} }
    - S_{\mathbf{X} \mathbf{Z} }
    { S_{\mathbf{Z}}}^{-1} 
    S_{\mathbf{Z}  \mathbf{Y}}.
\end{align*}
 Then the following formulas for
 conditional variances and covariances hold:
\begin{align*}
     \mbox{Var} [ \mathbf{X}_t \vert 
 \{ \mathbf{Z}_t \} ] & = \langle S_{ \mathbf{X} |\mathbf{Z}} \rangle, \\
  \mbox{Var} [ \mathbf{Y}_t \vert 
 \{ \mathbf{Z}_t \} ] & =
 \langle S_{ \mathbf{Y} |\mathbf{Z}} \rangle, \\
 \mbox{Cov} [ \mathbf{X}_t, 
   \mathbf{Y}_t \vert \{ \mathbf{Z}_t \} ] 
   & = \langle S_{\mathbf{X} \mathbf{Y} |\mathbf{Z}} \rangle.
\end{align*}
 Moreover, the scalar quantity 
\[
 1-\frac{ \det \left[  \langle 
 S_{\mathbf{X} \mathbf{Y} |\mathbf{Z}} \rangle {\langle  S_{\mathbf{Y}|\mathbf{Z}} \rangle }^{-1}  \langle  S_{\mathbf{Y} \mathbf{X} 
 |\mathbf{Z} } \rangle \right]}{ \det \langle S_{\mathbf{X}|\mathbf{Z}}\rangle  }
\]
takes values in $[0,1]$ if
 $\langle S_{X | Z} \rangle $ is positive definite.
\end{prop}

\begin{proof}
Let $E [ \mathbf{X}_t \vert \{ \mathbf{Z}_t \} ]$
denote the optimal linear predictor of
$\mathbf{X}_t$ given the whole process
$ \{ \mathbf{Z}_t \} $.  Then this can be
expressed as $\Pi (B) Z_t$ for some 
filter $\Pi (B)$ with frequency response 
function $\Pi (z) = S_{\mathbf{X} \mathbf{Z}}
(\lambda) { S_{ \mathbf{Z}} (\lambda) }^{-1}$
by Theorem 8.3.1 of \cite{Brillinger}.
It follows that the residual process
$\mathbf{X}_t - E [ \mathbf{X}_t \vert \{ \mathbf{Z}_t \} ]$ is stationary with 
spectral density 
\begin{align*}
S_{ \mathbf{X} \vert
  \mathbf{Z}  } (\lambda) = & S_{\mathbf{X} } (\lambda)
- \Pi (z) S_{ \mathbf{Z}\mathbf{X}} (\lambda)
- S_{ \mathbf{X}\mathbf{Z}} (\lambda) 
{ \Pi (z) }^{*}
+ \Pi (z) S_{ \mathbf{Z}} (\lambda) 
{ \Pi (z) }^* \\
& = S_{\mathbf{X} } (\lambda) - 
 S_{\mathbf{X} \mathbf{Z}}
(\lambda) { S_{ \mathbf{Z}} (\lambda) }^{-1}
S_{ \mathbf{Z}\mathbf{X}} (\lambda).    
\end{align*}
The residual process 
$\mathbf{Y}_t - E [ \mathbf{Y}_t \vert \{ \mathbf{Z}_t \} ]$ has an analogous
expression for its spectral density,
and the cross-spectral density between
the two residual processes is
\[
S_{ \mathbf{X} \mathbf{Y} \vert
  \mathbf{Z}  } (\lambda) = S_{\mathbf{X} \mathbf{Y}} (\lambda) - 
 S_{\mathbf{X} \mathbf{Z}}
(\lambda) { S_{ \mathbf{Z}} (\lambda) }^{-1}
S_{ \mathbf{Z}\mathbf{Y}} (\lambda). 
\]
Since the marginal variance of a  stationary process 
 is the average integral of its  spectral density,  the stated  variance and covariance formulas
 follow at once. Therefore we obtain  
\[
 \mbox{Var} [ \mathbf{X}_t \vert \mathbf{Y}_t, \{ \mathbf{Z}_t \} ]
=   \langle S_{\mathbf{X}|\mathbf{Z}}\rangle -
  \langle 
 S_{\mathbf{X} \mathbf{Y} |\mathbf{Z}} \rangle {\langle  S_{\mathbf{Y}|\mathbf{Z}} \rangle }^{-1}  \langle  S_{\mathbf{Y} \mathbf{X} 
 |\mathbf{Z} } \rangle,   
\] 
which is a non-negative definite matrix.
For any positive semi-definite matrices $A$ and $B$  of the same dimension, if $A-B \geq \mathbf{0}$ (i.e., the difference is non-negative definite), then $ \det A \geq \det B.$  Thus, setting $A = \mbox{Var} [ \mathbf{X}_t \vert \{\mathbf{Z}_t\}]$ and $B= \mbox{Var} [ \mathbf{X}_t \vert \{\mathbf{Z}_t\}] - \mbox{Var} [ \mathbf{X}_t \vert \{\mathbf{Y}_t\}, \{\mathbf{Z}_t\}]$ we find
that 
\[
\det \left[  \langle 
 S_{\mathbf{X} \mathbf{Y} |\mathbf{Z}} \rangle {\langle  S_{\mathbf{Y}|\mathbf{Z}} \rangle }^{-1}  \langle  S_{\mathbf{Y} \mathbf{X} 
 |\mathbf{Z} } \rangle \right] \leq  \det \langle S_{\mathbf{X}|\mathbf{Z}}\rangle,
\] 
 and the stated result follows.
\end{proof}

As an application of the above discussion,
we now consider $\{ \mathbf{Y}_t \}$ generated by
a linear filter-based privacy mechanism
$\Psi (B)$, i.e.,
\[ 
\mathbf{Y}_t = \Psi(B)X_t = \sum_j \Psi_j X_{t-j}. 
\]
Such a $\{ \mathbf{Y}_t \}$ clearly
satisfies the conditions of Proposition
\ref{prop:prvcy_measure}, and therefore
facilitates the following privacy definition.

\begin{Definition} [$\mlip$]
Let $\{\mathbf{X}_t, \mathbf{Z}_t\}$ be jointly stationary multivariate time series with spectral matrix \eqref{eq:specmat}, and positive definite  Schur complement $S_{ \mathbf{X} \vert 
 \mathbf{Z} }$. 
Then the {\it multivariate Linear Incremental Privacy} ($\mlip$) of $\{\mathbf{X}_t\}$ given $\{\mathbf{Z}_t\}$  with respect to the linear filtering mechanism $\Psi$ is defined as 
\beqn \lb{m-LIP}
\mlip(\Psi, S_{\mathbf{X} |\mathbf{Z}})  
 =  1-\frac{ \det \left[ \langle S_{\mathbf{X}|\mathbf{Z}}, \Psi \rangle
 \langle  \Psi S_{\mathbf{X}|\mathbf{Z}}, \Psi \rangle^{-1}\langle \Psi,  S_{\mathbf{X}|\mathbf{Z}}\rangle \right]}{ \det \langle S_{\mathbf{X}|\mathbf{Z}}\rangle  }.
\eeqn
\end{Definition}

Note that $\mlip$ is  a multivariate extension of LIP; see \cite{MRH2023}. Observing that 
$ S_{\mathbf{X} \mathbf{Y} |\mathbf{Z}} (\lambda)
=  S_{\mathbf{X}  |\mathbf{Z}} (\lambda)
 { \Psi (z) }^*$, it follows from  Proposition \ref{prop:prvcy_measure}, because $S_{ \mathbf{X} \vert  \mathbf{Z} }$ is positive definite,
 that  $\mlip$ takes values in $[0,1]$.  The value of zero occurs when  $\mbox{Var} [ \mathbf{X}_t \vert \mathbf{Y}_t, \{ \mathbf{Z}_t \} ]$ is singular,
 corresponding to complete predictability  of $ \mathbf{X}_t$ on the basis of  $ \mathbf{Y}_t$ and $\{  \mathbf{Z}_t \}$;  since $S_{ \mathbf{X} \vert 
 \mathbf{Z} }$ is positive definite,  it follows that  $\mbox{Var} [ \mathbf{X}_t \vert \{ \mathbf{Z}_t \} ]$ is non-singular,  so that the culprit in disclosing  $\mathbf{X}_t$ is $\mathbf{Y}_t$,  and not $\{ \mathbf{Z}_t \}$.  On the other  hand, when $\mlip$ equals one
 it must be the case that  $\langle S_{\mathbf{X}|\mathbf{Z}}, \Psi \rangle
 \langle  \Psi S_{\mathbf{X}|\mathbf{Z}}, \Psi \rangle^{-1}\langle \Psi,  S_{\mathbf{X}|\mathbf{Z}}\rangle$  is singular, i.e., that  $ \mbox{Var} [ \mathbf{X}_t \vert 
 \{ \mathbf{Z}_t \} ] 
 - \mbox{Var} [ \mathbf{X}_t \vert \mathbf{Y}_t, \{ \mathbf{Z}_t \} ]$ is singular.  This means that $\mathbf{Y}_t$ incurs  no additional ability to predict  certain linear combinations of $\mathbf{X}_t$ over and above what is
 already furnished by $\{ \mathbf{Z}_t \}$.

 \subsection{Privacy-Utility Optimization}
We present a framework for constructing a privacy mechanism -- denoted as $\Psi$ -- that possesses favorable privacy and utility characteristics. In the context of the augury solution, any $S_{\mathbf{X}}$-MAP filter $\Psi$ guarantees perfect second-order utility. Consequently, the selection of $\Psi$ should primarily align with the minimum privacy requirements. In particular, we seek an ``optimal'' $\Psi$ to maximize the privacy metric $\mlip(\Psi, S_{\mathbf{X}|\mathbf{Z}})$:
\be
\Psi_{opt}  = \underset{\Psi} 
{\arg\max} \;\mlip (\Psi, S_{\mathbf{X} |\mathbf{Z}}).
\ee
The optimization is over the class of $\sxmap$ filters. Given that the objective function is a nonlinear non-convex function of the filter, the optimization is rendered feasible by narrowing the class of all-pass filters. We use the parameterized class $\mcF_S$ in \eqref{eq:s_map} as the set over which the  objective function is optimized. 
Thus, given a conditional spectral density $S_{\bX| \bZ}$, the optimal filter is defined as 
\be
\lb{eq:mlip-opt}
\Psi_{opt}
= \underset{\Psi \in \mcF_{S_{\bX|\bZ}}} 
{\arg\min}\frac{ \det \left[ \langle  S_{\mathbf{X}|\mathbf{Z}}, \Psi \rangle
\langle  \Psi S_{\mathbf{X}|\mathbf{Z}}, \Psi \rangle^{-1}\langle \Psi, S_{\mathbf{X}|\mathbf{Z}}\rangle \right]}{ \det \langle S_{\mathbf{X}|\mathbf{Z}}\rangle  }.   
\ee
Given that the $\smap$ filters in $\mcF_S$ are defined with respect to unitary matrices, the optimization effectively reduces to a search over the set of unitary operators $U(z)$. Consequently, parameterizing unitary operators via their cepstral representation \eqref{cepstral unitary}, we can perform the optimization over the Euclidean space. 

 \section{Feasible Implementation of m-LIP}
 \lb{spectral factorization}

In practice, selection of an optimal $\smap$ filter according to \eqref{eq:mlip-opt} is based upon a spectral density $S$ estimated from  the available data (or based on prior knowledge). To use the class of $\smap$ filters in \eqref{eq:s_map} one needs to obtain square roots of a positive definite spectral density. Thus, the spectral density 
estimation procedure   must  constrain the estimator to be positive definite. Subsequent to the estimation of the spectral density, the spectral square root factors $S^+$ need to be computed at each frequency. Then the optimal $\smap$ filter is obtained using optimization of the criterion \eqref{eq:mlip-opt} over the parametric class \eqref{eq:s_map}   defined based on the estimated spectral factor. Finally, the filter coefficients associated with the optimal filter need to be computed using the inverse Fourier transform of the filter. The following section describes the step-by-step process of implementing the $\mlip$ privacy mechanism to a a given data set consisting (after removal of smooth trend) of the multiple time series of interest $\{ \bX_t \}$ and a set of auxiliary time series $\{ \bZ_t \}$.  

\subsection{Positive Definite Estimation of Spectral Densities}
For implementation of the $\mlip$ via spectral density estimation it is imperative that ${\hat{S}}_{\mathbf{X,Z}}$'s -- and hence the Schur complement  ${\hat{S}}_{\mathbf{X}|\mathbf{Z}}$ -- be  positive definite. In particular, with nonparametric approaches we must be careful to ensure this positive definite property is exhibited in the spectral density estimate almost surely.

Any such spectral estimator yields a  
$\hat{S}_{\mathbf{X}}$-MAP filter rather than
a ${S}_{\mathbf{X}}$-MAP filter, and thus
there will be some degradation of second-order
utility due to statistical estimation error
of the spectral density; this is different
from the univariate case explored in 
\cite{MRH2023}, wherein an all-pass filter
can be constructed without knowing the 
spectral density of the input process.
However, it can be argued that the practical utility that practitioners care about is
based on the finite sample at hand, and
the preservation of {\it sample} autocovariances, i.e., 
$\hat{\Gamma}_{\mathbf X} (h) = 
\hat{\Gamma}_{\mathbf Y} (h)$ for all 
$h \in \mathrm{Z}$.  Such a ``sample'' -- or
feasible -- 
second-order utility is equivalent to
$\{ \mathbf{X}_t \}$ and $\{ \mathbf{Y}_t \}$
having the same periodogram. Hence, setting $\hat{S}_{\mathbf{X}}$ to be the periodogram would guarantee feasible second-order utility, but
unfortunately the multivariate periodogram
is a rank one matrix for all $\lambda$, and
hence violates our positive definite 
requirement.  Therefore, we recognize there
may be some feasible loss of sample utility   due to positive definite spectral density
estimation; however, as sample size increases
these estimates will be consistent for the true $\hat{S}_{\mathbf{X}, \mathbf{Z}}$, as will the sample autocovariances for the process' autocovariances, and thus for large sample sizes second-order utility will approximately hold.

Given detrended data $\{\bW_t\} = \{\bX_t, \bZ_t\}$, there are several different options for obtaining positive definite spectral density estimates. One option is to fit a parametric model, such as an order $p$ vector autoregressive process (or VAR($p$)), and use the spectral density of that model evaluated at the estimated parameters. Another option is to use a non-parametric estimator that is constrained to be positive definite. In this article, we use the non-parametric kernel estimator of $S_{\mathbf{X,Z}}$ proposed in \cite{Politis2011}.  In \cite{Politis2011}, the author uses a flat-top kernel because it is an infinite-order kernel, and therefore is capable of achieving higher-order accuracy. 
The disadvantage of flat-top kernels is that they are not necessarily positive semi-definite. For this reason, the author lets $\epsilon_T>0$ be some chosen sequence decreasing to zero as $T \to \infty$, and truncates the eigenvalues of the flat-top taper estimator to $[\epsilon_T, \infty).$ 

We choose $\epsilon_T=1/T$ here and employ the flat-top taper method on the sample autocovariances to get a positive definite (PD) estimator. Let ${\hat{S}}_{\bX, \bZ}(\lambda)$ be the flat-top taper PD estimator of the residual spectral density obtained using $\epsilon_T=1/T$ for a sample of size $T.$ The top left block of the estimator will be denoted as ${\hat{S}}_{\bX}$, and is the PD estimator of $S_{\bX}$, and the Schur complement ${\hat{S}}_{\bX |\bZ} $ will be the estimator of the residual spectral density.

\subsection{Spectral Factorization}
The multivariate spectral factorization problem is  fundamental  in spectral analysis, wherein the objective is to obtain a vector moving average (VMA) representation of order q that corresponds to a given set of  autocovariances, denoted as $\Gamma(0),\dots, \Gamma(q-1), \Gamma(q)$. The requirement is that 
$\sum_{|h| \leq q} \Gamma(h)e^{-i\lambda h}$
must be positive definite for all values of the frequency parameter $\lambda \in [-\pi, \pi]$. 

There are several available methods for spectral factorization; we follow the method of  Bauer
 \cite{Bauer1955}, as
 summarized in \cite{McElroyJTSA2018}.
 First, we approximate the
 spectral density $S(\lambda)$ by
 $\sum_{|h| \leq q} \Gamma(h)e^{-i\lambda h}$
 for $q$ large; for simplicity of exposition,
 suppose this holds exactly, i.e.,
\[
S(\lambda)= \sum_{h=-q}^{q} \Gamma(h) e^{-i\lambda h}.
\]
Bauer's method first forms the block Toeplitz
covariance matrix of a time series sample
of length $m$ (where $m$ is taken as large
as computationally feasible), and secondly the 
modified Cholesky decomposition (MCD) 
is computed.  The lower left block row
of the Cholesky factor consists
(as $m \longrightarrow \infty$) of the 
autocovariances $\Gamma (q), \Gamma (q-1),
 \ldots, \Gamma (0)$, as described in
 \cite{McElroyJTSA2018}.  
Then the spectral factorization can be concisely represented as
\[
S(z)= S^{+}(z)\, { S^{+}(z) }^* = \Theta(z)\, \Sigma \,{ \Theta(z) }^*,
\]
where the spectral factor $S^{+}(z)$ assumes the form $S^{+}(z)= \Theta(z) \, \Sigma^{1/2}$.
Here $\Theta( B)=
\sum_{k=0}^{q}\Theta_k B^k$ is an order $q$ matrix polynomial in $B$
such that $\Theta_0 = \bI$,
and whose coefficients are the VMA
coefficients.  Also, $\Sigma$ is   the covariance matrix of the innovations.  
The spectral factor ${\hat{S}}_{\bX}^{+} (\lambda)$ obtained from using the Bauer algorithm on the flat-top taper PD estimator ${\hat{S}}_{\bX}(\lambda)$ is used in the design of $\smap$ filters.

\subsection{Parameterization of the $\smap$ Class}

Once the estimated spectral factor ${\hat{S}}_{\bX}^{+} (\lambda)$ has been obtained, one can construct the parametric class $\mcF_S$ of $\smap$ filters given in \eqref{eq:s_map}
by setting $S = {\hat{S}}_{\bX}^{+} (z).$
The free parameters of the class are obtained from the matrices $\Omega_k$ in the cepstral representation  $U(z) = \exp\{\Omega(z)\} = \exp\{\sum_{k \in \ZZ} \Omega_k z^k \}$
of the unitary operator. 
We can parameterize $\Omega(z)$ by allowing
the matrix entries of $\Omega_{k}$ for $k > 0$
to be any real number, and for $k < 0$
we set $\Omega_{k} = -\Omega_{-k}^{\prime}$.
For $k =0$, we only need to constrain 
$\Omega_{0}$ to be skew-symmetric, which is 
achieved by freely parameterizing the lower
triangular portion of the matrix, and enforcing
that the upper triangular portion to be equal
to the negative transpose of the lower portion
(and the diagonal entries are zero). 
For feasible implementation, we need to truncate the Laurent series $\Omega (z)$ at a finite stage, say $r$. Thus, the class of filters $\Psi(z)$ that we are choosing to optimize over are of the form  
\be
\Psi_r(z) = {\hat{S}}_{\bX}^{+} (\lambda)\exp\{\sum_{k= -r}^r \Omega_k z^k\} {\hat{S}}_{\bX}^{+} (\lambda)^{-1}, 
\lb{eq:smap_r}
\ee
where $\Omega_{-k} = -\Omega_k^T$ for all $k \geq 0.$
The truncation stage $r$ has to be chosen by the data curator, and can be done by examining the optimal privacy value for several different choices of $r$.  Given $r$, the number of free parameters in the class is $n_r = rn^2 + \binom{n}{2}$, which is linear in the cepstral length $r$ and quadratic in $n$.

\subsection{Optimal All-pass Filter Selection}
\lb{all-pass optim} 
In view of the filters described in \eqref{eq:smap_r},
the criterion (\ref{eq:mlip-opt}) can be optimized with respect to the $n_r$ free parameters in $\Omega_0, \Omega_1, \ldots, \Omega_r.$
However, the complicated nature of the $\mlip$ objective function precludes an analytical solution,
and we instead proceed via non-linear optimization techniques.

Our numerical method leverages an optimization algorithm known as AGMsDR \cite{AGMsDR} that is suitable for nonlinear nonconvex optimization.
While conventional optimization techniques like Brent or L-BFGS typically yield dependable results, our preference for AGMsDR stems from its specialized capability to address non-convex and non-smooth functions. Although our objective function is not inherently non-smooth, its non-convex nature makes the AGMsDR algorithm particularly attractive. Additionally, this method proves valuable in situations where  more commonly employed methods may encounter convergence issues.

Consider the cepstral series $\Omega (z)$
truncated to some order $r$, so that
\[
\Omega (z) = \sum_{k=-r}^r \Omega_k z^k
 = \Omega_0 + \sum_{k=1}^r \Omega_k z^k
  - \sum_{k=1}^r \Omega_k^{\prime} z^{-k}.
\]
Let $\vartheta$ denote the vector of $n_r$ real parameters corresponding to the  entries of the cepstral matrices $\Omega_k$ for $k= 0, 1, \ldots, r$. The unitary  operator $U(z)$ then becomes a function of the free parameters, and we denote it as $U(z;\vartheta).$  Also, let $\Psi_r(z; \vartheta) =  S^{+}_\mathbf{X} \, U(z; \vartheta) \,{ S^{+}_\mathbf{X} }^{-1}.$
Then the solution to the optimal filter problem (\ref{eq:mlip-opt}) can be  re-expressed as
\begin{equation} 
\label{eq:mlip-crit-final}
\vartheta_{opt} =  \underset{\vartheta} 
{\arg\min} \frac{ \det \left[ \langle  S_{\mathbf{X}|\mathbf{Z}}, 
 \Psi_r(z; \vartheta) \rangle
\langle \Psi_r(z; \vartheta) S_{\mathbf{X}|\mathbf{Z}}, \Psi_r(z; \vartheta) \rangle^{-1}\langle  \Psi_r(z; \vartheta), S_{\mathbf{X}|\mathbf{Z}}\rangle \right]}{ \det \langle S_{\mathbf{X}|\mathbf{Z}}\rangle  },
\end{equation}
with $\Psi_{opt}(z) =  \Psi_r(z; \vartheta_{opt}).$
For initialization of the $\vartheta$ parameters  we draw a random sample of size $n_r$ from the standard normal distribution, and set the initial values equal to the obtained sample. 
After the optimal filter $\Psi_{opt}(z)$ has been determined, the filter coefficients are obtained by Fourier inversion: $\Psi_k = \langle z^{-k} , \Psi_{opt}(z) \rangle$.

\subsection{Estimation of  Trend and Forecast Extension}
 Before the application of the estimated filter to the data, the deterministic trend needs to be estimated and removed from the multiple time series. Then   after the application of the filter, the estimated trend is added back to the privatized times series. 

Trend estimation can be done using different available software. For this article, we used the differencing method to achieve the detrended series using the $diff()$ function (details in Section \ref{data analysis}). After the removal of trends from each of the series, we obtain the detrended data, which is then used for filtering. The filter is two-sided and of finite length, say $M$ on each side. To get a series with the same length as the original data after filtering, We extend the detrended series by $M$ time points on each side by using one-sided forecasts. Since we are assuming that the spectral density is known for the original series, we use this same spectral density  to generate optimum one-sided $h-$ ahead forecasts for $h - 0, 1, \ldots, M.$ After we obtain the filtered series by applying the filter  $\smap$ to the detrended series, we add back the estimated trends. A privatized series with a trend is thereby generated.

\subsection{ Realized Utility}

Due to the error that occurred during spectral estimation, and due to finite sample effects, there can be utility loss; we measure this loss through the Frobenius norm, which for a complex matrix A is defined via $\parallel A \parallel=\sqrt{\mbox{tr} (AA^*)}$. 
The  Frobenius Discrepancy (FD) (see \cite{McElroyRoy2021})
of the two $n$-variate spectral density matrices $S_{\mathbf{X}}$ and $S_{\mathbf{Y}}$
is the average (over frequencies) of
the squared Frobenius norm of their difference, viz.
\[
\text{FD}(S_{\mathbf{X}},S_{\mathbf{Y}})= \langle \parallel S_{\mathbf{X}}-S_{\mathbf{Y}}\parallel^2\rangle.
\]
A property of FD is
\[
\text{FD}(S_{\mathbf{X}},S_{\mathbf{Y}})=0 \; \text{if and only if} \; S_{\mathbf{X}} \overset{a.e.}{=} S_{\mathbf{Y}},
\]
where ``a.e.'' indicates that the two matrix-valued functions are equal at all frequencies
$\lambda \in [-\pi, \pi]$ except for a subset
of Lebesgue measure zero. The above property is referred to as the complete equivalency of $S_{\mathbf{X}}$ and $S_{\mathbf{Y}}$; since the discrepancy of the two spectral densities on a set of measure zero does not disrupt the equality of their corresponding autocovariances, it follows that complete equivalency entails second-order utility.  

Another expression for $\text{FD}(S_{\mathbf{X}},S_{\mathbf{Y}})$ is
$\sum_{h \in \mathrm{Z} } \parallel \Gamma_\mathbf{X} (h) - \Gamma_\mathbf{Y} (h) \parallel^2$, which makes the connection
to second-order utility more explicit.
When using FD to assess second-order utility
(low values corresponding to higher utility),
 it is convenient to use a normalized measure;
 to that end, we derive the upper bound
\[
\text{FD}(S_{\mathbf{X}}, S_{\mathbf{Y}})
\leq \sum_{h \in \mathrm{Z}}
 { \left( \| \Gamma_{\mathbf{X}}(h)\| +
  \| \Gamma_{\mathbf{Y}}(h) \|  \right) }^2.
\]
This is obtained using the triangle inequality
for the Frobenius norm.  
We use this upper bound to normalize the Frobenius discrepancy, obtaining the
so-called NFD:
\beqn
\lb{NFD}
\text{NFD}(S_{\mathbf{X}},S_{\mathbf{Y}})= \frac{\sum_{h \in \mathrm{Z}} 
\parallel \Gamma_\mathbf{X} (h) - \Gamma_\mathbf{Y} (h) \parallel^2 }{ 
\sum_{h \in \mathrm{Z}}
 { \left( \| \Gamma_{\mathbf{X}}(h)\| +
  \| \Gamma_{\mathbf{Y}}(h) \|  \right) }^2 }.
\eeqn
For two matrices $A$ and $B$, the bound
$\| A- B \| \leq \| A \| + \| B \|$
is achieved for $B = -A$, 
 which indicates that the maximum value
 of NFD is $1$.
An empirical version of NFD, denoted as
$\widehat{\text{NFD}}$, is obtained by substituting sample autocovariances in
(\ref{NFD}).  Finally, We
define the realized utility measure (RUM) via
\beqn
\lb{RUM}
\text{RUM}(S_{\mathbf{X}},S_{\mathbf{Y}})=1-\widehat{\text{NFD}}(S_{\mathbf{X}},S_{\mathbf{Y}}),
\eeqn
which has the property that high values
(close to unity) correspond to high utility
(i.e., when the FD is close to zero).
 Also, because NFD is bounded by one,
 low values of RUM correspond to low utility.

\section{Numerical Illustration}
\lb{numerical illustration}
In this section we apply the multivariate LIP methods to both simulated data and real data --
the QWI employment data published by U.S. Census Bureau. 

\subsection{Simulated Data}
Here we simulate data from a Vector Autoregressive Moving Average (VARMA) process of order (1,1), a VAR(1) with i.i.d. innovations, and a VAR(1) where the innovations are drawn from an Autoregressive
Conditionally Heteroscedastic (ARCH) process of
order 1 (for detailed discussion of VARMA and ARCH models, see \cite{brockwelldavis}).  These simulation
processes are used to jointly describe 
$\{ \mathbf{X}_t \}$ and $\mathbf{Z}_t$;
for the third case, the ARCH(1) innovations  correspond
to $\{ \mathbf{Z}_t \}$.

For obtaining the privatization filter $\Psi (B)$ in each case, we employ the following settings. For spectral density matrix estimation, we use the flat-top taper method described above. We obtain the spectral factorization for the joint spectral density of the target series that are the focal point of our protective measures. We then solve the minimization problem posed in (\ref{eq:mlip-crit-final}), using various choices of the order $r$ of 
$\Omega (z)$.  

When  $\Omega(z)$ equals the zero matrix $\mathbf{0}$, corresponding to $U(z) = \bI$, then the $\mlip$ criterion equals zero -- which makes sense since no privatization actually occurs. The choice  $r=0$ means that only $\Omega_0$ is present, and there is only $\binom{n}{2} = 1$ parameter -- the single lower triangular entry -- in $\vartheta$. Secondly,  $r = 1$ yields $4$ free parameters in $\Omega_1$, plus one free parameter in $\Omega_0$.  A third scenario keeps three of the four elements of $\Omega_1$ constant so that $\vartheta$ consists of two parameters -- one for $\Omega_0$, and one corresponding to the free parameter in $\Omega_1$. In each of these three scenarios, we minimize the criterion to obtain the optimal $\vartheta$ and the corresponding filter $\Psi (B)$.
 
 We plot the histograms of the realized privacy values for the VAR(1) and VARMA(1,1) simulation, and for those plots, we set $r=0$. For each of the three cases, we plot the comparisons of the autocorrelation and the cross-correlation functions of the original and the released series. For those plots, we use $r=1$ to obtain the optimal 
 $\sxmap$   filter.

\subsubsection{Simulation from VAR(1)}
 Here we describe the chosen parameter values for the simulation. 
 A VAR(p) model for $ \{ \mathbf{W}_t \}$ is defined as follows:
\[
\mathbf{W}_t = A_1 \mathbf{W}_{t-1} + A_2 \mathbf{W}_{t-2} + \ldots + A_p \mathbf{W}_{t-p} + \mathbf{\varepsilon}_t, 
\]
where  $A_1, A_2, \ldots, A_p$ are coefficient matrices for lags 1 through $p$, and $ \{ \varepsilon_t \} \sim \text{WN}(0, \Sigma)$.

We generate a time series of length $T=2000$ from a 4-variate VAR(1) model. The AR coefficient matrix is 
\[
A= \begin{pmatrix}
    0.5 & 0.1 & 0.0 & 0.0 \\
0.2 & 0.4 & 0.1 & 0.0 \\
0.1 & 0.2 & 0.6 & 0.2 \\
0.0 & 0.1 & 0.2 & 0.5
\end{pmatrix},
\]
which has all absolute eigenvalues   less than 1,
thereby ensuring stationarity and causality
of the process. The covariance matrix of the noise is assumed to be 
\[
\Sigma = \begin{pmatrix}
    1.0 & 0.2 & 0.1 & 0.0 \\
0.2 & 1.0 & 0.2 & 0.1\\
0.1 & 0.2 & 1.0 & 0.3\\
0.0 & 0.1 & 0.3 & 1.0
\end{pmatrix}.
\]

We divide  the 4-variate VAR(1) process into
two parts: 
the first two components correspond to
$\{ \mathbf{X}_t \}$, 
while the latter two components correspond
to $\{ \mathbf{Z}_t \}$.
Generating the process in this fashion serves the purpose of keeping the $\{ \mathbf{X}_t \}$ and $\{ \mathbf{Z}_t \}$ time series jointly stationary.

We generate the 4-dimensional VAR(1) time series multiple times (100 Monte Carlo copies), and obtain   optimal values of $\vartheta$ for various sample lengths and instances. We measure the time complexities, and report the average time complexity for each case. For the two-parameter and five-parameter cases the 
privacy filter resulted in maximal privacy
for almost all simulations.
In the one-parameter case ($r=0$) the privacy
measure was not clustered tightly around unity,
and we  report the histogram
in 
Figure \ref{histogram}.  To demonstrate
utility, we plot sample autocovariances
$\hat{\Gamma}_{ \mathbf{X}} (h)$ and
$\hat{\Gamma}_{ \mathbf{Y}} (h)$ for a
single simulation in Figure 
\ref{fig:original_vs_filtered_series_acvf_VAR(1)}, when $r=1$. 

\begin{figure}[t]
\centering
\includegraphics[width = 2.2in, height = 1.6in]{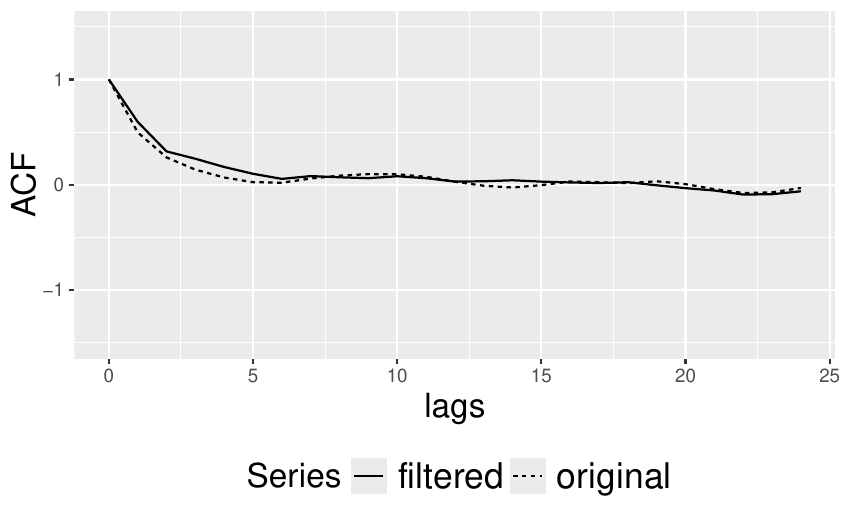}
\includegraphics[width = 2.2in, height = 1.6in]{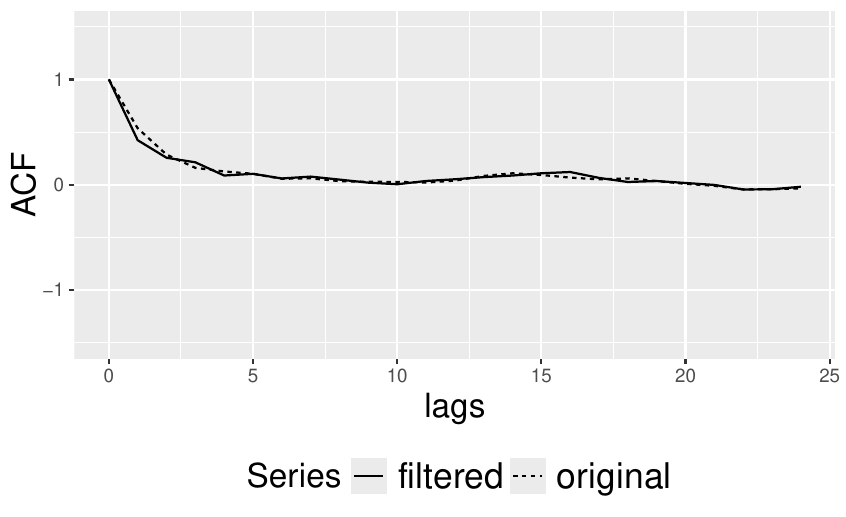}\\
\includegraphics[width = 2.2in, height = 1.6in]{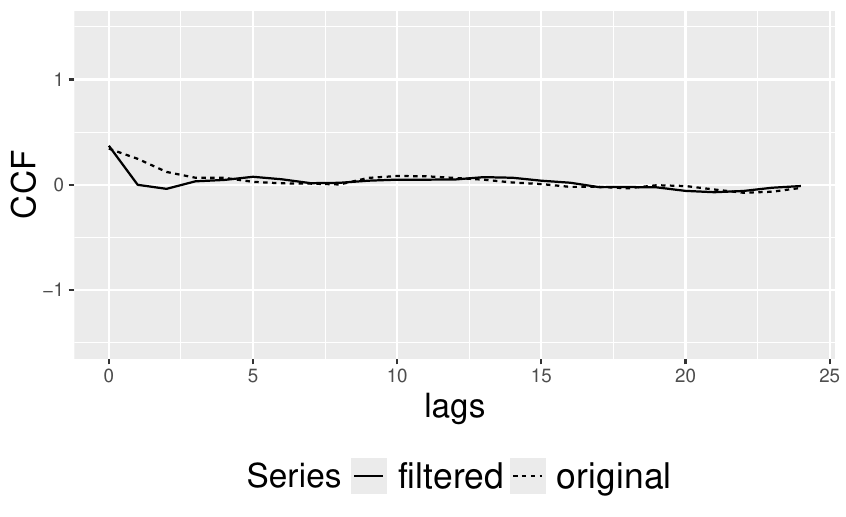}

\caption{Comparison of sample autocorrelation function (ACF) and the cross-correlation function (CCF) of the original and the filtered copies for the first and second series for the case $r=1$ (VAR(1)). The top row shows the two ACF plots while the bottom plot shows the CCF between the two series. 
\lb{fig:original_vs_filtered_series_acvf_VAR(1)} }
\end{figure}

\begin{figure}[h]
\centering
\includegraphics[width = 2.5in, height = 1.8in]{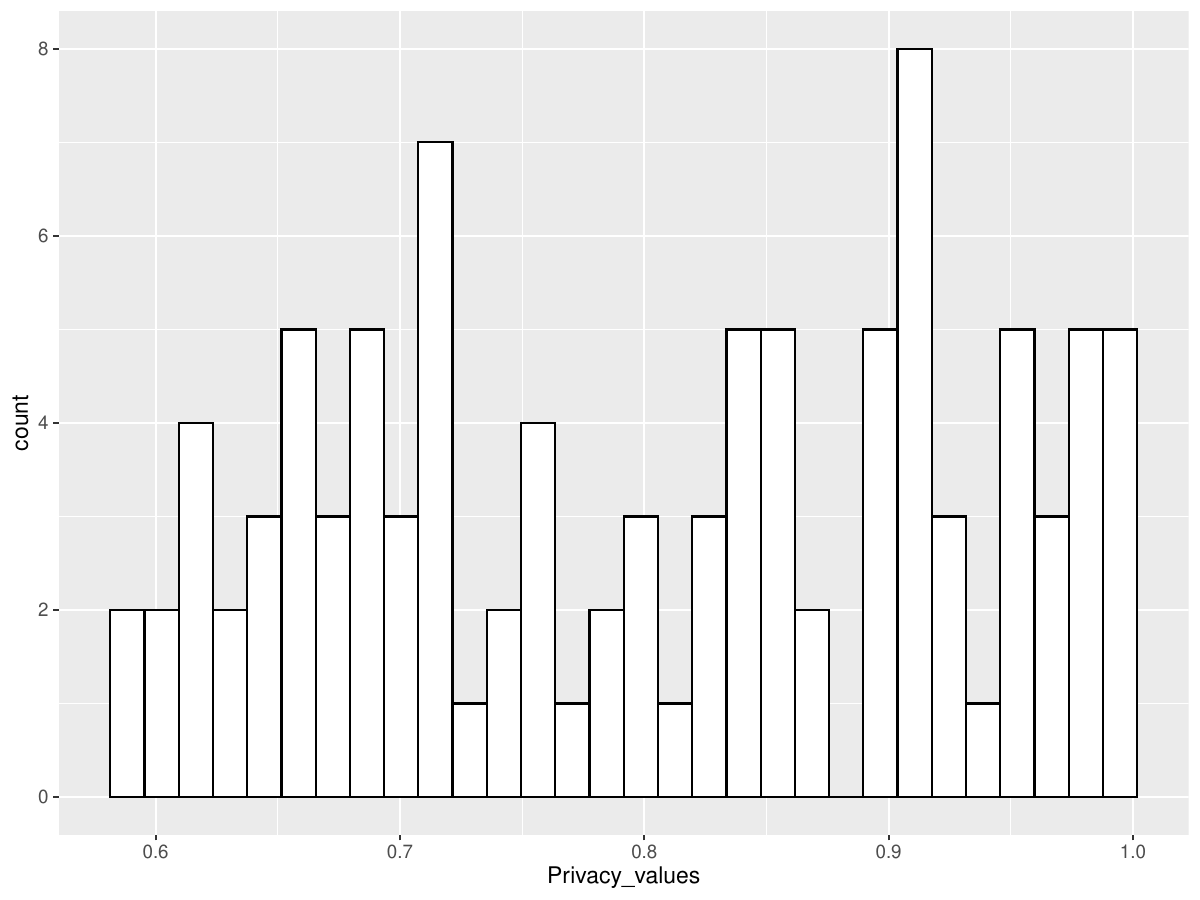}
\includegraphics[width = 2.5in, height = 1.8in]{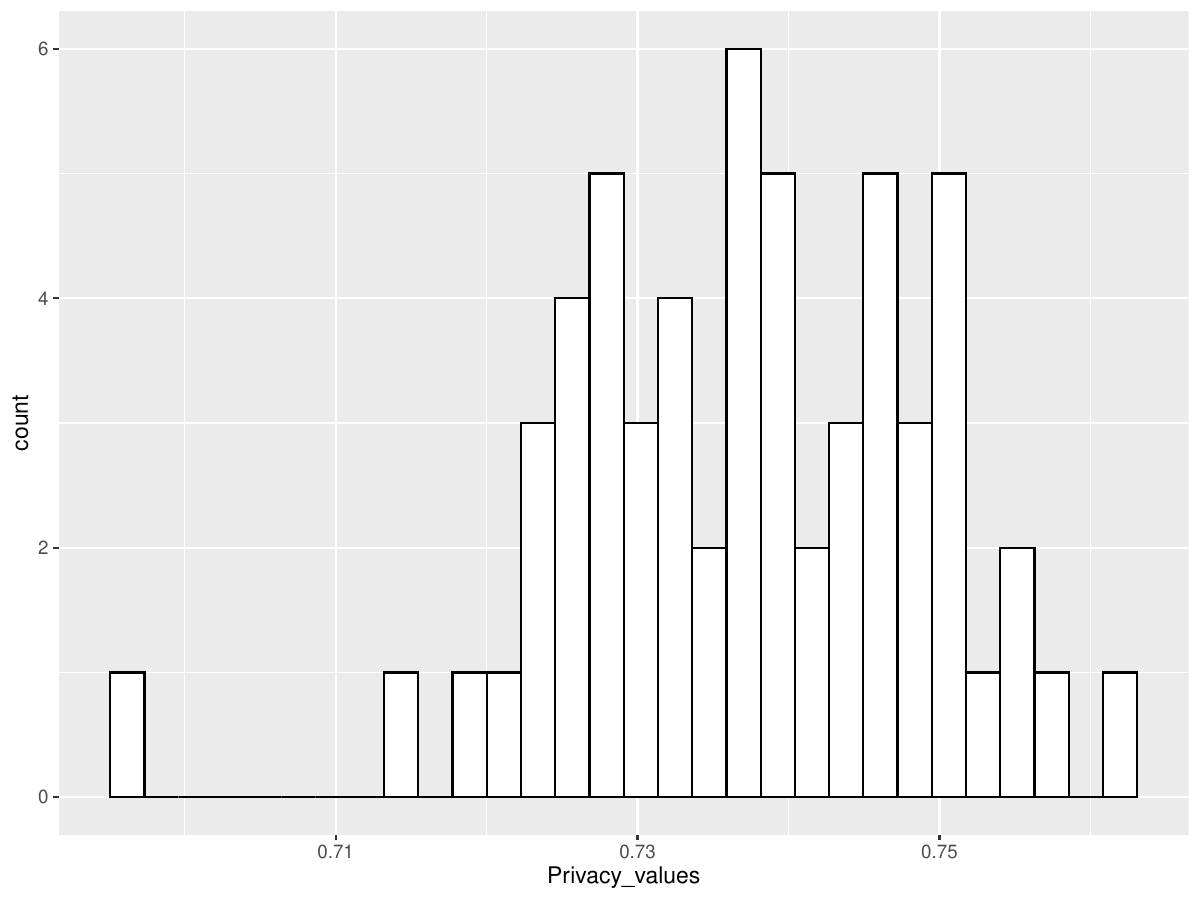}\\

\caption{Histograms of $\mlip$ values
 for VAR(1) (left) and VARMA(1,1) (right),
 $r=0$ case.
\lb{histogram} }
\end{figure}

\subsubsection{Simulation from VARMA(1,1)}
We generate a 4-variate VARMA(1,1) described by the following equation:
\[
\mathbf{W}_t = \Phi \mathbf{W}_{t-1} + \mathbf{\epsilon}_t + \Theta \mathbf{\epsilon}_{t-1}, 
\]
where $\{\mathbf{\epsilon}_t\}$ is a white noise process with innovation variance-covariance matrix 
\[ 
\Sigma = \begin{pmatrix}
    0.09 & 0 & 0 & 0\\
    0 & 0.03 & 0 & 0\\
    0 & 0 & 0.05 & 0\\
    0 & 0 & 0 & 0.07
\end{pmatrix}.
\] 
The coefficient matrices for the Autoregressive (AR) and Moving Average (MA) components are defined respectively as 
\[
\Phi = \begin{pmatrix}
   -0.00556 &-0.6353 &0.2529& -0.0096\\
   -0.2288  &0.3506 &0.2414 &-0.02505\\
   -0.23423 &-1.33007& 0.517& -0.1978\\
 0.1624 & 0.5523& 0.4042& -0.1412
\end{pmatrix}
\]
and 
\[
\Theta = 
\begin{pmatrix}
  \begin{matrix}
  0.6 & 0.2 \\
  0 & 0.3
  \end{matrix}
  & \rvline & \bigzero \\
\hline
  \bigzero & \rvline &
  \begin{matrix}
  0 & 0 \\
  0 & 0
  \end{matrix}
\end{pmatrix}.
\]
Both $\{ \mathbf{X}_t \}$ and $\{ \mathbf{Z}_t \}$
are defined from the VARMA(1,1) process in the same
manner as in the previous simulation.
We also construct our privatization filter using the same settings, and assess performance in the same way. In Figure \ref{fig:original_vs_filtered_series_acvf_VARMA}, we present a comparison of the sample 
  autocovariances
$\hat{\Gamma}_{ \mathbf{X}} (h)$ and
$\hat{\Gamma}_{ \mathbf{Y}} (h)$ for a
single simulation, for the case $r=1$. In Figure \ref{histogram} we present the histograms for the case $r=0$.

\begin{figure}[ht]
\centering
\includegraphics[width = 2.2in, height = 1.6in]{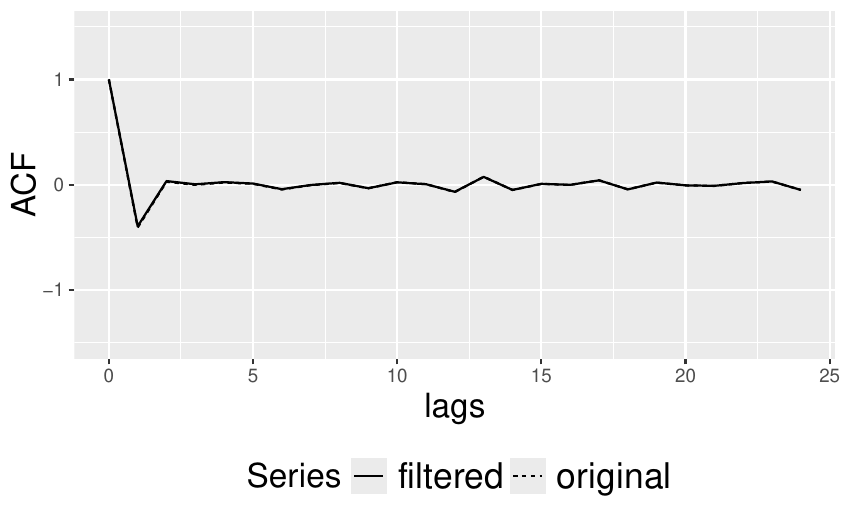}
\includegraphics[width = 2.2in, height = 1.6in]{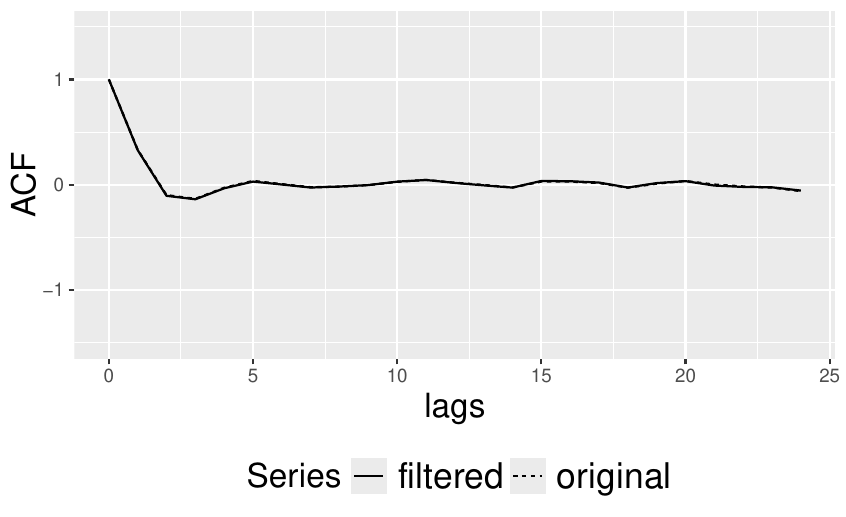}\\
\includegraphics[width = 2.2in, height = 1.6in]{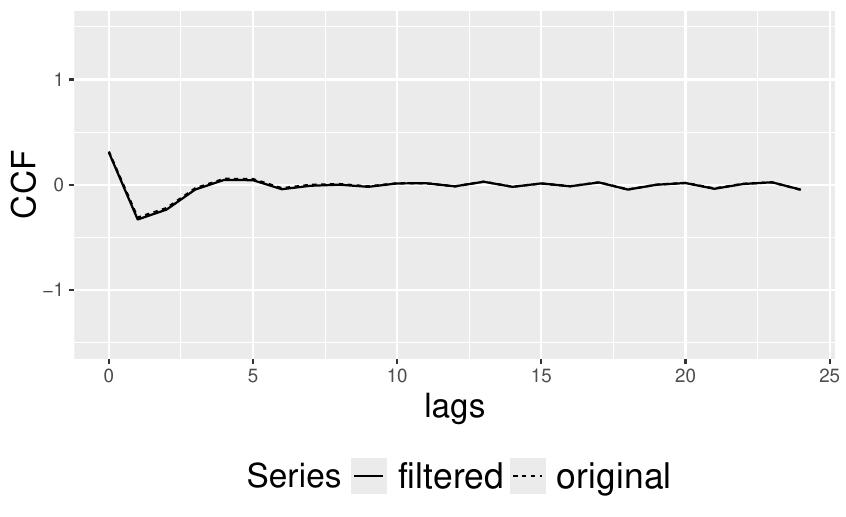}

\caption{Comparison of sample autocorrelation and the cross-correlation functions of the original and the filtered copies for the first and second series for the case $r=1$ (VARMA(1,1)). The top row shows the two ACF plots while the bottom plot shows the CCF between the two series. 
\lb{fig:original_vs_filtered_series_acvf_VARMA} }
\end{figure}

\subsubsection{Simulation of a VAR(1) with ARCH(1) errors}

We generate a bivariate VAR(1) following the equation
\[
\mathbf{Q}_t = A_1 \mathbf{Q}_{t-1}+ \mathbf{\zeta}_t,
\]
where $\mathbf{\zeta}_{t,1}=\sqrt{h_t}e_t$ and $e_t \sim$ i.i.d. standard normal.  $\mathbf{\zeta}_{t,1}$ stands for the first component of the innovation series $\mathbf{\zeta}_t$, i.e. $\mathbf{\zeta}_t=(\mathbf{\zeta}_{t,1}, \mathbf{\zeta}_{t,2})'$. Here, $h_t$ is defined by 
\[
h_t = \alpha_0 + \alpha_1 \mathbf{\zeta}^2_{t-1,1}.
\]
For our simulation we set $\alpha_0=1$ and $\alpha_1=0.5$. We assume $\mathbf{\zeta}_{t,2} \sim \mbox{WN} (0,1)$ and is drawn independently with respect to $\mathbf{\zeta}_{t,1}$. The series $\{\mathbf{Q}_t\}$ serves in the role of $\{\mathbf{X}_t\}$, where $\{\mathbf{Z}_t\}$ is the first component of the $\{\mathbf{\zeta}_t\}$. The autocorrelation comparison is plotted in Figure \ref{fig:original_vs_filtered_series_acvf_VARARCH}.

\begin{figure}[ht]
\centering
\includegraphics[width = 2.2in, height = 1.6in]{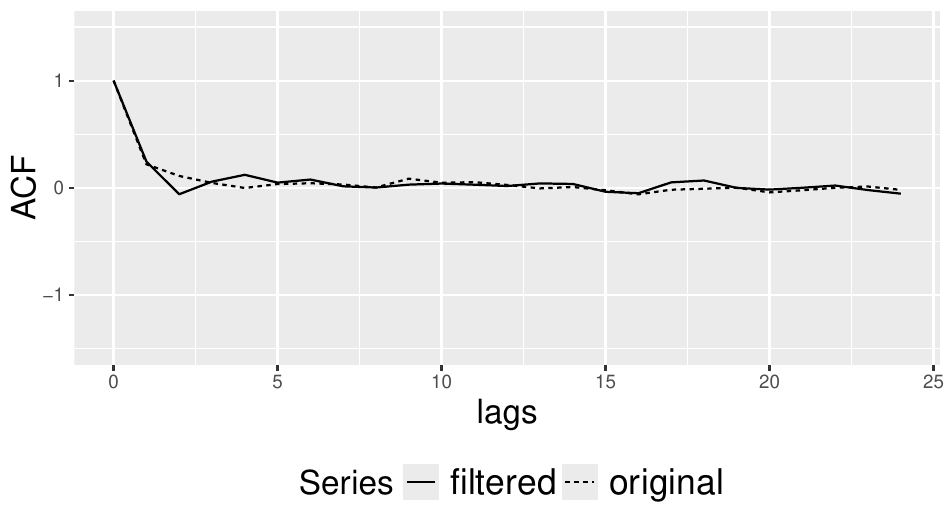}
\includegraphics[width = 2.2in, height = 1.6in]{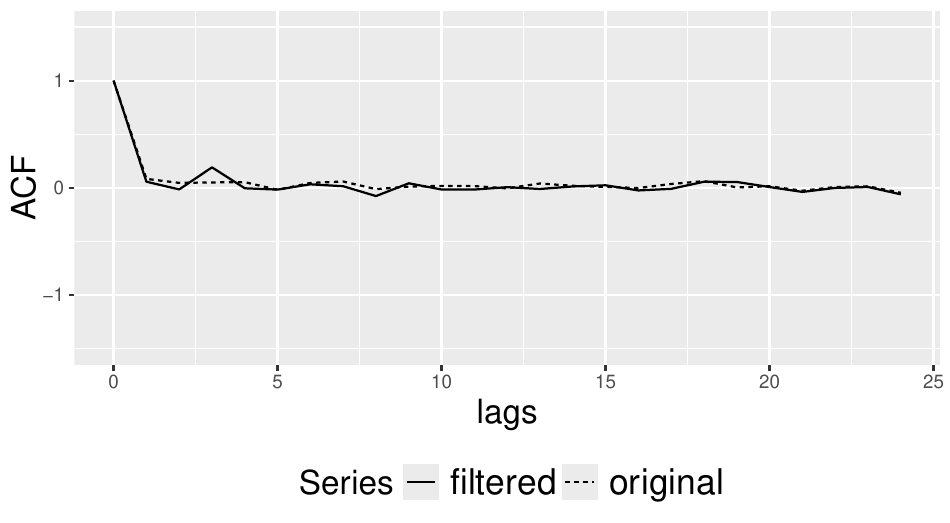}\\
\includegraphics[width = 2.2in, height = 1.6in]{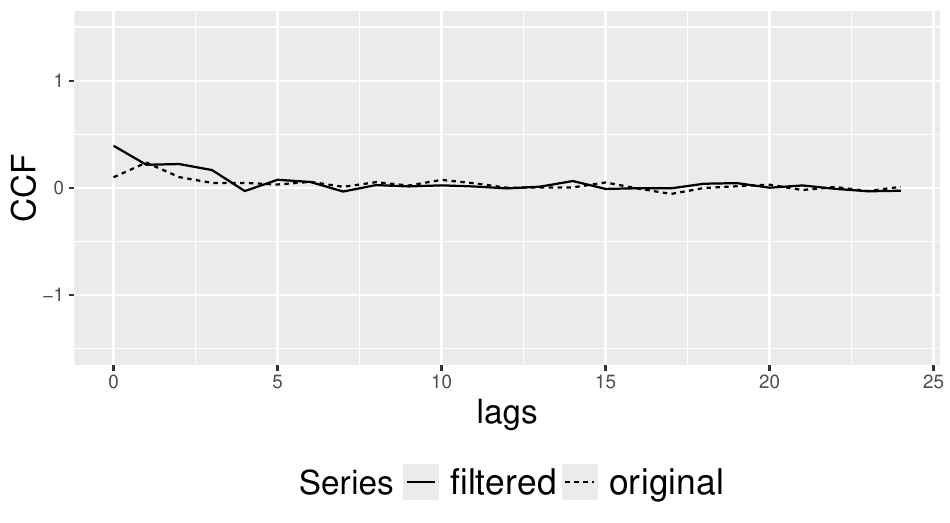}

\caption{Comparison of sample autocorrelation and the cross-correlation functions of the original and the filtered copies for the first and second series for the case $r=1$ (VAR(1), ARCH(1) error). The top row shows the two ACF plots while the bottom plot shows the CCF between the two series. 
\lb{fig:original_vs_filtered_series_acvf_VARARCH} }
\end{figure}

\begin{table}
\begin{adjustbox}{width=\columnwidth,center}
\begin{tabular}{|*{10}{c|}}  

\hline
\multicolumn{1}{|c}{} & \multicolumn{3}{|c}{VAR(1)} & \multicolumn{3}{|c}{VARMA(1,1)} & \multicolumn{3}{|c|}{VAR(1),ARCH(1) error} \\ \hline
Parameters & \multicolumn{1}{|c}{Min. Privacy} & \multicolumn{1}{|c}{Max. Privacy} & \multicolumn{1}{|c}{Time (avg)} & 
\multicolumn{1}{|c}{Min. Privacy} & \multicolumn{1}{|c}{Max. Privacy} & \multicolumn{1}{|c}{Time (avg)} & 
\multicolumn{1}{|c}{Min. Privacy} & \multicolumn{1}{|c}{Max. Privacy} & \multicolumn{1}{|c|}{Time (avg)}  \\ \hline
1 & 0.692 &  0.998 & 22.7  &  0.695 &  0.85 & 91.54 & 0.652 &  0.894 & 82.75              \\ \hline
2 & 0.786 & 0.9989 & 48.9 & 0.895 & 0.9991 & 89.391 & 0.7973 & 1 & 122.431                 \\ \hline
5 & 0.996 & 1& 339.76  & 0.9921 & 1& 432.567 & 0.9254 & 1& 323.698  \\ \hline
\end{tabular}
\end{adjustbox}
\caption{Privacy values and time complexities.
\lb{table1} }
\end{table}

\begin{table}
\begin{adjustbox}{width=\columnwidth,center}
\begin{tabular}{ |p{2cm}|p{1.5cm}|p{2.5cm}| p{5.5cm}| }
 \hline
 \multicolumn{4}{|c|}{$T=2000$} \\
 \hline
 Parameters & VAR(1) & VARMA(1,1) & VAR(1) with ARCH(1) error\\
 \hline
 1   &   0.968 & 0.92 &   0.976\\
 2 &   0.902 & 0.975   & 0.908\\
 5 & 0.935 & 0.999&  0.901\\
 \hline
\end{tabular}
\end{adjustbox}
\caption{Realized Utility Measure.
\lb{table:utility measure} }
\end{table}

\subsubsection{Comparison of The Three Simulations}
For each of the aforementioned three cases we generate 100 Monte Carlo copies of the coefficient series. For each of those instances we obtain privacy values. In Table \ref{table1}
we report the minimum privacy value, maximum maximum privacy value and the time taken on average for each of the cases (VAR(1), VARMA(1) and VAR(1) with ARCH(1) error, $T=2000$) for different number of parameters (so $\vartheta$ has
length $1$, $2$, or $5$). The average privacy value (average taken over the Monte Carlo simulations) for the VAR(1) when $r=0$ is $0.802$; for VARMA(1,1) the average privacy is $0.737$,  and for the third simulation it is $0.758$. The average maximum privacy value for 
$r=1$ (5 parameters) for the VAR(1) simulation is $0.9978$, whereas for the VARMA(1,1) it is $0.998$, and is  $0.9967$ in the third case. We also display the utility values defined in \ref{RUM} for different number of parameters in Table \ref{table:utility measure}.

\subsection{QWI Employment Data}
\lb{data analysis}

In our data analysis, we demonstrate the effectiveness of our method by utilizing employment count data obtained from the Quarterly Workforce Indicators (QWI) dataset published by the U.S. Census Bureau. The QWI dataset is derived from a comprehensive collection of job and work location administrative records spanning 49 states, and it is updated quarterly; see
\cite{abowd2011national} for full details on the data's construction and publication.

All data used in our analysis were retrieved from the QWI Explorer website \cite{QWI}
on January 28, 2024, at 10:00 pm. Our analysis centers on the quarterly indicator referred to as ``Beginning of Quarter Employment: Count,'' which we will abbreviate as ``employment count.'' The dataset covers the state of Maryland and spans from the first quarter of 1997 (Q1 1997) to the fourth quarter of 2022 (Q4 2022). Specifically, we have gathered data for four distinct counties within Maryland: Baltimore, Frederick, Montgomery, and Howard counties. 


Our objective is to safeguard the bivariate time series comprising employment counts for Baltimore and Frederick counties, with Montgomery and Howard counties constituting the series that may be known to potential attackers. The employment data spanning 26 years from the aforementioned four counties in Maryland are visualized in Figure \ref{fig:QWI Employment Count}.

\begin{figure}[htbp!]
\centering
\includegraphics[width = 2.8in, height = 2in]{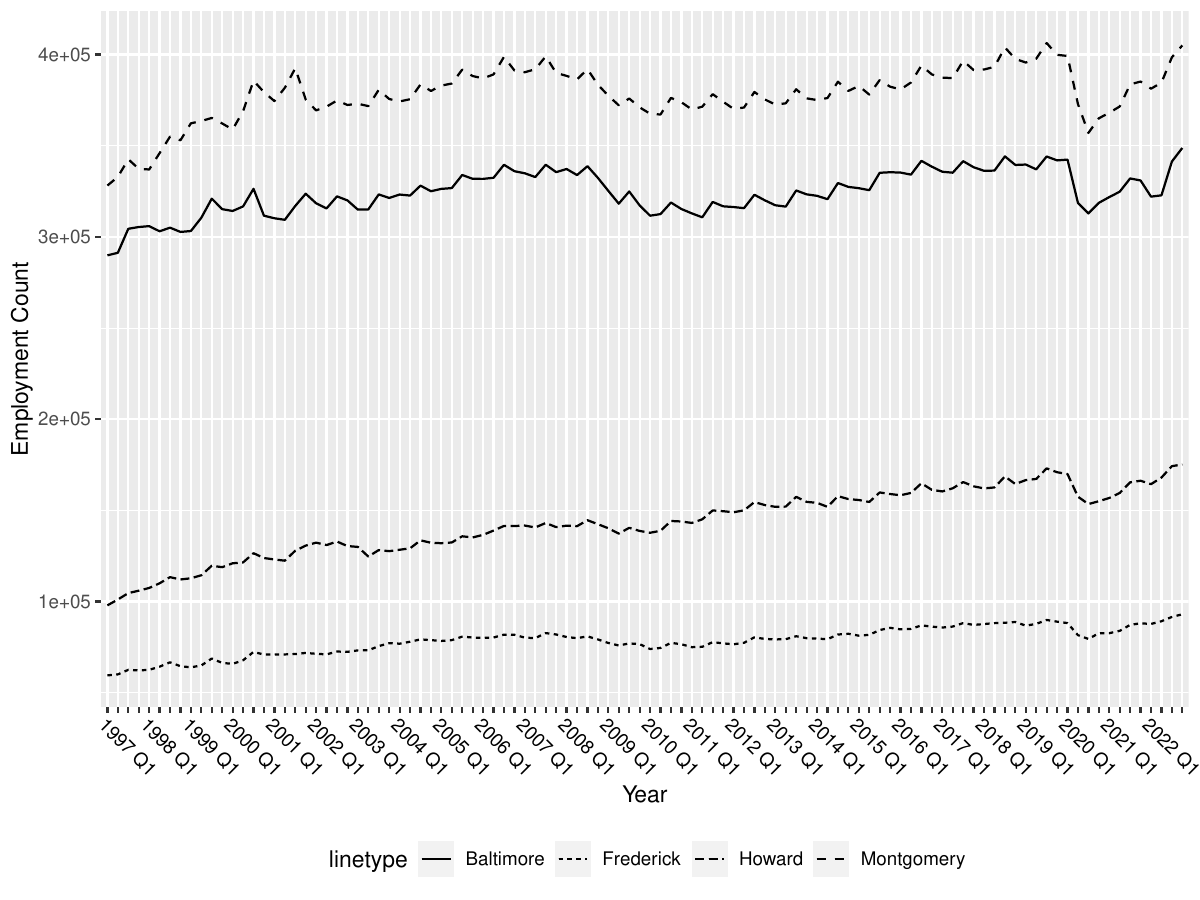}
\caption{QWI employment count for Maryland counties. 
\lb{fig:QWI Employment Count} }
\end{figure}


We remove trend and seasonal patterns 
from the quarterly data by
applying the seasonal differencing
operator $1- B^4$.  The resulting 
`annual growth rate' time series is stationary,
as is verified through visual inspection of the autocorrelation function and the application of
the augmented Dickey-Fuller test on each of the time series.


\begin{figure}[h]
\centering
\includegraphics[width = 2.2in, height = 1.6in]{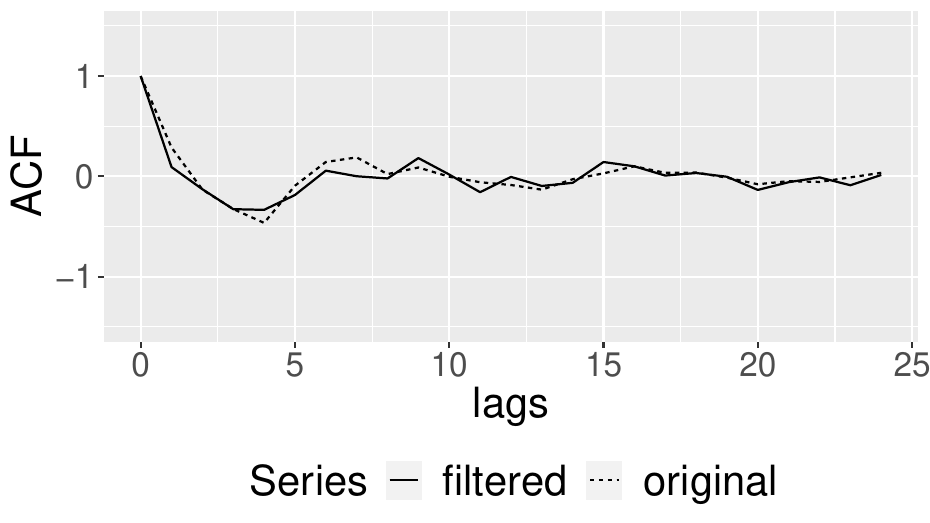}
\includegraphics[width = 2.2in, height = 1.6in]{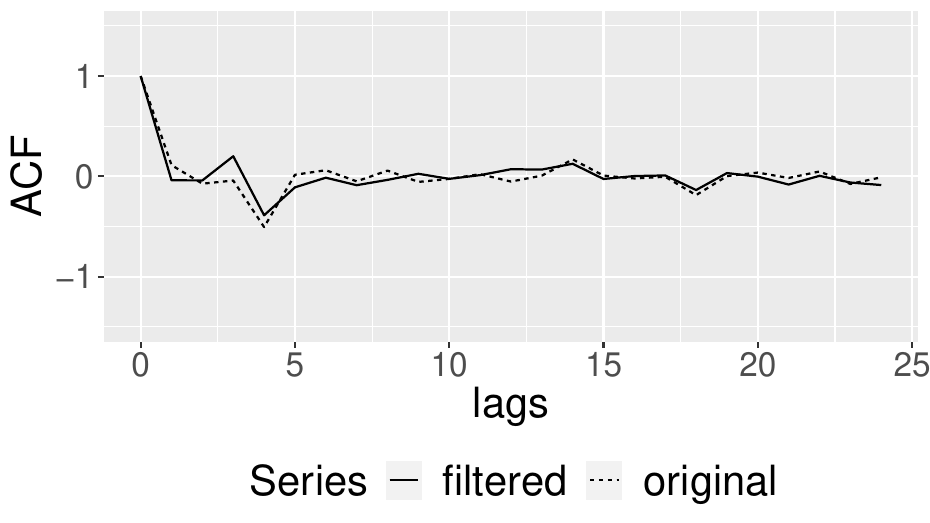}\\
\includegraphics[width = 2.2in, height = 1.6in]{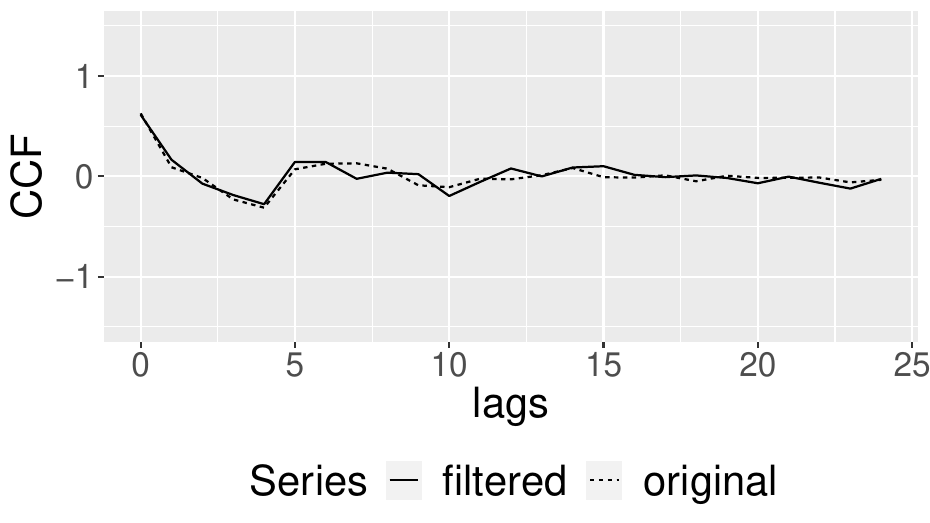}

\caption{Comparison of sample autocorrelation and the cross-correlation functions of the original and the filtered copies for the first and second detrended series for the QWI data. The top row shows the two ACF plots while the bottom plot shows the CCF between the two series. 
\lb{fig:QWI_original_vs_filtered_series_acvf} }
\end{figure}

\begin{figure}[!htbp]
\centering
\includegraphics[width = 3.5in]{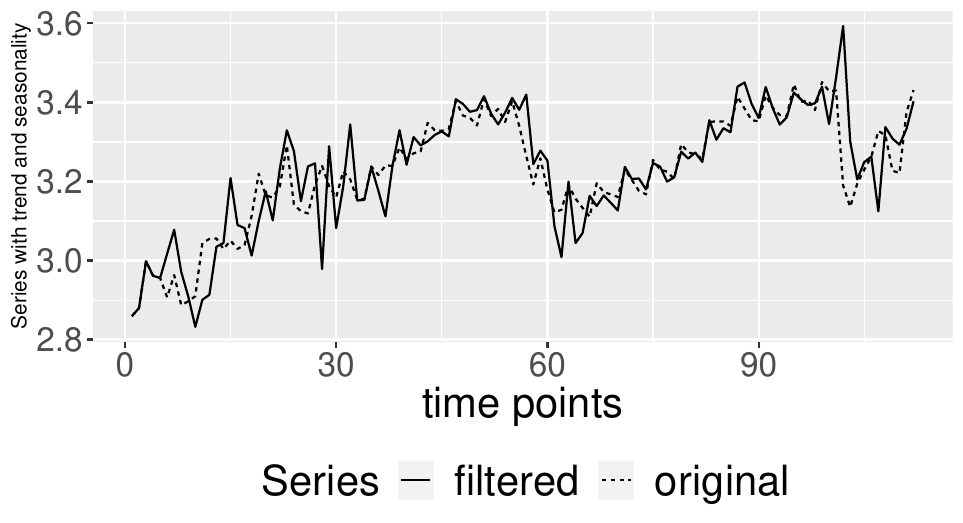}
\includegraphics[width = 3.5in]{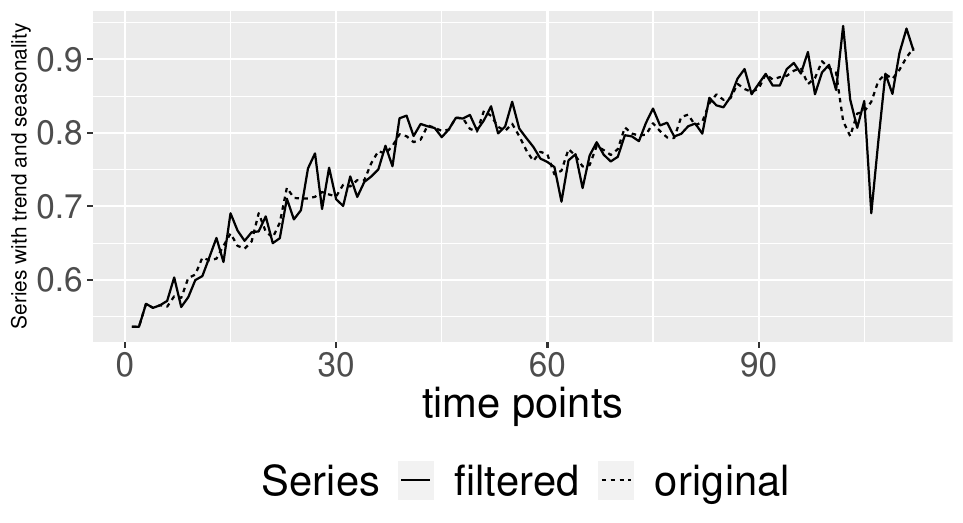}

\caption{Comparison of standardized sample paths of the original and the filtered copies for the Baltimore (top panel) and Frederick (bottom panel) series. The $y$-axis is employment count (in units of $10^5$).} 
\lb{fig:QWI_original_vs_filtered_series_sample_paths} 
\end{figure}

We obtain an $\sxmap$ filter with the choice $r=1$, and apply the filter to the growth rate data to get the privatized growth rate series. 
Then we recursively determine modified data
in the original scale, inverting the action
of the $1-B^4$ filter. The sample paths for Baltimore County and Frederick County, along with their corresponding filtered counterparts, are displayed in Figure~\ref{fig:QWI_original_vs_filtered_series_sample_paths}. The comparisons of autocovariance and cross-covariance series are depicted in Figure~\ref{fig:QWI_original_vs_filtered_series_acvf}. 

The plots in Figure~\ref{fig:QWI_original_vs_filtered_series_acvf} show us that the autocorrelation structure of the two series are successfully kept unaltered, preserving utility. Moreover,  
the cross-correlations are preserved as well -- a feature that is not available in current 
univariate privacy mechanisms. From the two plots in Figure~\ref{fig:QWI_original_vs_filtered_series_sample_paths} it is apparent that the sample paths of the actual series and the released series coincide very rarely, and yet the released series  maintains the trend and seasonal structure of the original data. 
Thus, the released time series serves as a representative proxy for the original time series, striking a balance between privacy and utility.

\section{Discussion and Future Work}

In this paper, we propose a novel privacy preservation technique for multivariate time series, denoted as $\mlip$, which leverages the concept of multivariate all-pass filtering. Multivariate all-pass filtering represents a more intricate approach compared to its univariate counterpart and relies on the spectral density matrix of the target series requiring protection. 

The effectiveness of our proposed method hinges upon the prerequisite of stationarity in the underlying series. In this paper, we have implemented the multivariate mechanism after removing deterministic trends from each component. Thus, the proposed implementation is a two-stage procedure that suffers from the drawbacks of multi-stage methods, where errors from previous stages can influence the outcome of subsequent stages. A single-stage implementation that constrains multivariate all-pass filters capable of accommodating $d$th order polynomial trends in the model is more desirable. Such procedures would exclude the macro trends from the privacy budgets, and thereby leave them invariant under the implementation of the multivariate mechanism. 

A linear filter $\Psi$ whose application leaves a $d$th order polynomial unchanged can be found by constraining $ \Psi (z)$ so as to ensure it is trend-invariant.  For $d=0$ (the case of a
constant trend) it is necessary that
$\Psi (1) =  \Psi(e^{-i\lambda}) \vert_{\lambda = 0}$ equals the identity matrix.  For $d > 0$,
it is required that the $d$th derivative of
$\Psi(e^{-i\lambda})$ (with respect to $\lambda$)
at $\lambda = 0$ is the zero matrix.  
In \cite{MRH2023} such conditions on 
the filter were parsed in terms of conditions on the cepstral coefficients.  However, in the multivariate case the derivative of
$\exp \{ \Omega (e^{-i \lambda}) \}$ is not easy to compute, due to the fact that the summands $\Omega_k e^{-i \lambda k}$ do not commute with one another.  Hence, we cannot directly impose trend-invariant filter constraints on $\Psi (z)$ through conditions on $\vartheta$. This poses a formidable challenge. We intend to explore methods for choosing MAP filters that pass polynomial trends unchanged as a topic of future investigation. 

In some applications, it may be reasonable to include the macro features such as trend and seasonality in the privacy budget. For example, if one series has a strikingly different trend, or unique seasonal pattern, it may require disclosure avoidance. We plan to investigate privacy mechanisms applicable to such situations in the future.

\bibliographystyle{plain}
\bibliography{references}  

\end{document}

%% file: abbreviation.tex
\newtheorem{Result}{Result}

\newtheorem{Definition}{Definition}

\newcommand{\be}{\begin{equation}}
\newcommand{\ee}{\end{equation}}
\newcommand{\beqa}{\begin{eqnarray*}}
\newcommand{\eeqa}{\end{eqnarray*}}
\newcommand{\beqn}{\begin{eqnarray}}
\newcommand{\eeqn}{\end{eqnarray}}
\newcommand{\baa}{\begin{array}}
\newcommand{\eaa}{\end{array}}
\newcommand{\bcc}{\begin{center}}
\newcommand{\ecc}{\end{center}}
\newcommand{\btab}{\begin{tabular}}
\newcommand{\etab}{\end{tabular}}

\newcommand{\lb}{\label}

\newcommand{\mcF}{\mathcal{F}}

\newcommand{\bA}{\boldsymbol{A}}
\newcommand{\bB}{\boldsymbol{B}}

\newcommand{\bI}{\boldsymbol{I}}

\newcommand{\bW}{\boldsymbol{W}}

\newcommand{\bU}{\boldsymbol{U}}

\newcommand{\bX}{\boldsymbol{X}}

\newcommand{\bY}{\boldsymbol{Y}}
\newcommand{\bZ}{\boldsymbol{Z}}

\newcommand{\CC}{{\mathbb C}}

\newcommand{\ZZ}{{\mathbb Z}}

\newcommand{\bigzero}{\mbox{\normalfont\Large\bfseries 0}}
\newcommand{\rvline}{\hspace*{-\arraycolsep}\vline\hspace*{-\arraycolsep}}

\newcommand*{\smap}{\ensuremath{S\textsf{-MAP}}}
\newcommand*{\sxmap}{\ensuremath{S_{\bX}\textsf{-MAP}}}
\newcommand*{\mlip}{\ensuremath{m\textsf{-LIP}}}

%% file: main.bbl
\begin{thebibliography}{10}

\bibitem{Abowd2012}
J.~M. Abowd, K.~Gittings, K.~L. McKinney, B.~E. Stephens, L.~Vilhuber, and S.~Woodcock.
\newblock Dynamically consistent noise infusion and partially synthetic data as confidentiality protection measures for related time series.
\newblock {\em US Census Bureau Center for Economic Studies Paper No. CES-WP-12-13, Available at SSRN: https://ssrn.com/abstract=2159800 or http://dx.doi.org/10.2139/ssrn.2159800}, 2012.

\bibitem{abowd2011national}
J.~M. Abowd and L.~Vilhuber.
\newblock National estimates of gross employment and job flows from the quarterly workforce indicators with demographic and industry detail.
\newblock {\em Journal of econometrics}, 161(1):82--99, 2011.

\bibitem{Abowd2020}
John~M Abowd, Robert Ashmead, Ryan Cumings-Menon, Simson Garfinkel, Micah Heineck, Christine Heiss, Robert Johns, Daniel Kifer, Philip Leclerc, Ashwin Machanavajjhala, et~al.
\newblock The 2020 census disclosure avoidance system topdown algorithm.
\newblock {\em Harvard Data Science Review}, (Special Issue 2), 2022.

\bibitem{Arcolezi2022}
H.~H. Arcolezi, J-F. Couchot, D.~Renaud, B.~Al Bouna, and X.~Xiao.
\newblock Differentially private multivariate time series forecasting of aggregated human mobility with deep learning: Input or gradient perturbation?
\newblock {\em Neural Computing and Applications}, 34:13355–13369, 2022.

\bibitem{Bauer1955}
F.~Bauer.
\newblock Ein direktes iterationsverfahren zur hurwitz-zerlegung eines polynoms.
\newblock {\em Archiv der elektrischen Übertragung}, 2017.

\bibitem{Brillinger}
D.~R. Brillinger.
\newblock {\em Time Series: Data Analysis and Theory - David R. Brillinger}.
\newblock Siam, 2001.

\bibitem{QWI}
U.S.~Census Bureau.
\newblock {Quarterly Workforce Indicator}.
\newblock \url{https://qwiexplorer.ces.census.gov}, 2023.
\newblock [Online; accessed in 2022 and 2023].

\bibitem{Dw2006}
C.~Dwork.
\newblock Differential privacy.
\newblock {\em International Colloquium on Automata, Languages and Programming, part II (ICALP)}, 2006.

\bibitem{DwMcNiSm2006}
C.~Dwork, F.~McSherry, K.~Nissim, and A.~Smith.
\newblock Calibrating noise to sensitivity in private data analysis.
\newblock {\em Theory of Cryptography Conference(TCC)}, pages 265--284, 2006.

\bibitem{DwRo2014}
C.~Dwork and A.~Roth.
\newblock The algorithmic foundations of differential privacy.
\newblock {\em Foundations and Trends in Theoretical Computer Science}, 9:211--407, 2014.

\bibitem{EFM2015}
M.~A. Erdogdu, N.~Fawaz, and A.~Montanari.
\newblock Privacy-utility tradeoff for time-series with application to smart-meter data.
\newblock {\em Association for the Advancement of Artificial Intelligence}, 2015.

\bibitem{Fior2019}
F.~Fioretto and P.~V. Hentenryck.
\newblock Optstream: Releasing time series privately.
\newblock {\em Journal of Artificial Intelligence Research}, 2019.

\bibitem{Meng2020}
Ruobin Gong and Xiao-Li Meng.
\newblock Congenial differential privacy under mandated disclosure.
\newblock FODS '20, page 59–70, New York, NY, USA, 2020. Association for Computing Machinery.

\bibitem{holan2017cepstral}
S.~Holan, T.~S. McElroy, and G.~Wu.
\newblock The cepstral model for multivariate time series: The vector exponential model.
\newblock {\em Statistica Sinica}, pages 23--42, 2017.

\bibitem{Hong2013}
S.K. Hong, K.~Gurjar, H.S. Kim, and Y.S. Moon.
\newblock A survey on privacy preserving time-series data mining.
\newblock {\em International Conference on Intelligent Computational Systems (ICICS)}, 2013.

\bibitem{Imtiaz2020}
Sana Imtiaz, Sonia-Florina Horchidan, Zainab Abbas, Muhammad Arsalan, Hassan~Nazeer Chaudhry, and Vladimir Vlassov.
\newblock Privacy preserving time-series forecasting of user health data streams.
\newblock In {\em 2020 IEEE International Conference on Big Data (Big Data)}, pages 3428--3437, 2020.

\bibitem{Kats2022}
M.~Katsomallos, K.~Tzompanaki, and D.~Kotzinos.
\newblock Landmark privacy: Configurable differential privacy protection for time series.
\newblock {\em Conference on Data and Application Security and Privacy (CODASPY)}, 2022.

\bibitem{Lako2021}
F.~L. Lako, P.~Lajoie-Mazenc, and M.~Laurent.
\newblock Privacy-preserving publication of time-series data in smart grid.
\newblock {\em Security and Communication Networks}, 2021.

\bibitem{Leukam2021}
Franklin Leukam, Paul Lajoie-Mazenc, and Maryline Laurent.
\newblock Privacy-preserving publication of time-series data in smart grid.
\newblock {\em Security and Communication Networks}, 2021:1--21, 2021.

\bibitem{Lyu2017}
L.~Lyu, Y.~W. Law, J.~Jin, and M.~Palaniswami.
\newblock Privacy-preserving aggregation of smart metering via transformation and encryption.
\newblock {\em IEEE Trustcom/BigDataSE/ICESS, pp. 472–479, IEEE, Sydney, Australia}, 2017.

\bibitem{MRH2023}
T.~McElroy, A.~Roy, and G.~Hore.
\newblock Flip: A utility preserving privacy mechanism for time series.
\newblock {\em Journal of Machine Learning Research}, 2023.

\bibitem{McElroyJTSA2018}
T.~S. McElroy.
\newblock Recursive computation for block-nested covariance matrices.
\newblock {\em Journal of Time Series Analysis}, 2017.

\bibitem{MP2020}
T.~S. McElroy and D.~N. Politis.
\newblock {\em Time Series: A First Course with Bootstrap Starter}.
\newblock CRC Press, 2020.

\bibitem{McElroyRoy2021}
T.~S McElroy and A.~Roy.
\newblock Model identification via total frobenius norm of multivariate spectra.
\newblock {\em Journal of the Royal Statistical Society Series B: Statistical Methodology}, 84(2):473--495, 2022.

\bibitem{AGMsDR}
Y.~Nesterov, A.~Gasnikov, S.~Guminov, and P.~Dvurechensky.
\newblock Primal–dual accelerated gradient methods with small-dimensional relaxation oracle.
\newblock {\em Optimization Methods and Software}, 36:773--810, 2021.

\bibitem{brockwelldavis}
R.~A.~Davis P.~J.~Brockwell.
\newblock {\em Introduction to Time Series and Forecasting}.
\newblock Springer.

\bibitem{Politis2011}
D.~N. Politis.
\newblock Higher-order accurate, positive semi-definite estimation of large-sample covariance and spectral density matrices.
\newblock {\em Econometric Theory}, 2011.

\bibitem{RN2010}
V.~Rastogi and S.~Nath.
\newblock Differentially private aggregation of distributed time-series with transformation and encryption.
\newblock {\em International Conference on Management of Data, ACM SIGMOD}, pages 735--746, 2010.

\bibitem{SST2009}
Y.~Sang, H.~Shen, and H.~Tian.
\newblock Privacy-preserving tuple matching in distributed databases.
\newblock {\em IEEE Transactions on Knowledge and Data Engineering, 21(12)}, page 1767–1782, 2009.

\bibitem{SCR2011}
E.~Shi, T-H.~H. Chan, and E.~Rieffel.
\newblock Privacy-preserving aggregation of time-series data.
\newblock {\em In Proc. of the Network and Distributed System Security Symposium, San Diego, California}, 2011.

\bibitem{SoCh2017}
S.~Song and K.~Chaudhuri.
\newblock Composition properties of inferential privacy for time-series data.
\newblock {\em arXiv:1707.02702}, 2017.

\bibitem{SoWaCh2017}
S.~Song, Y.~Wang, and K.~Chaudhuri.
\newblock Pufferfish privacy mechanisms for correlated data.
\newblock {\em arXiv:1603.03977}, 2017.

\bibitem{Stach2019}
C.~Stach.
\newblock Vault: A privacy approach towards high-utility time series data.
\newblock {\em International Conference on Emerging Security Information, Systems and Technologies, pp. 41–46}, 2019.

\bibitem{WaZh2009}
L~Wasserman and S.~Zhou.
\newblock A statistical framework for differential privacy.
\newblock {\em Journal of the American Statistical Association}, 105:375--389, 2009.

\bibitem{Acs2012}
G.~Ács, C.~Castelluccia, and R.~Chen.
\newblock Differentially private histogram publishing through lossy compression.
\newblock {\em IEEE International Conference on Data Mining}, 2012.

\end{thebibliography}
